\documentclass[runningheads]{llncs}
\usepackage{amsfonts} 
\usepackage[left=3cm, right=3 cm, top=3 cm]{geometry}
\usepackage{graphicx}
\usepackage[ruled,vlined]{algorithm2e}
\usepackage{amsmath}
\usepackage{subfigure}
\usepackage{marvosym}
\usepackage{skak}
\begin{document}
\title{Distributed Transformations of Hamiltonian Shapes based on Line Moves}
%

%
\author{Abdullah Almethen \and
Othon Michail \and
Igor Potapov}
%

%
\institute{Department of Computer Science, University of Liverpool, Liverpool, UK\\
\email{\{A.Almethen, Othon.Michail, Potapov\}@liverpool.ac.uk}}
\maketitle              
\begin{abstract}
We consider a discrete system of $n$ simple indistinguishable devices, called \emph{agents}, forming a \emph{connected} shape $S_I$ on a two-dimensional square grid. Agents are equipped with a linear-strength mechanism, called a \emph{line move}, by which an agent can push a whole line of consecutive agents in one of the four directions in a single time-step. We study the problem of transforming an initial shape $S_I$ into a given target shape  $S_F$ via a finite sequence of line moves in a distributed model, where each agent can observe the states of nearby agents in a Moore neighbourhood. 
Our main contribution is the first distributed connectivity-preserving transformation that exploits line moves  within a total of $O(n \log_2 n)$ moves, which is asymptotically equivalent to that of the best-known centralised transformations. The algorithm solves the \emph{line formation problem} that allows agents to form a final straight line $S_L$, starting from any shape $ S_I $, whose \emph{associated graph} contains a Hamiltonian path.
\end{abstract}
\section{Introduction}
The explosive growth of advanced technology over the last few decades has contributed significantly towards the development of a wide variety of distributed systems consisting of large collections of tiny robotic-units, known as \emph{monads}. These monads are able to move and communicate with each other by being equipped with microcontrollers, actuators and sensors. However, each monad is severely restricted and has limited computational capabilities, such as a constant memory and lack of global knowledge. Further, monads are typically homogeneous, anonymous and indistinguishable from each other. Through a simple set of rules and local actions, they collectively act as a single unit and carry out several complex tasks, such as transformations and explorations. 

In this context, scientists from different disciplines have made great efforts towards developing innovative, scalable and adaptive collective robotic systems. This vision has recently given rise to the area of programmable matter, first proposed by Toffoli and Margolus \cite{TM91} in 1991, referring to any kind of materials that can algorithmically change their physical properties, such as shape, colour, density and conductivity through transformations executed by an underlying program. This newborn area has been of growing interest lately both from a theoretical and a practical viewpoint.

One can categorise programmable matter systems into \emph{active} and \emph{passive}. Entities in the passive systems have no control over their movements. Instead, they move via interactions with the environment based on their own structural characteristics. Prominent examples of research on passive systems appear in the areas of population protocols \cite{AADFP06,MS16a,MS18}, DNA computing \cite{Ad94,BON96} and tile self-assembly \cite{Do12,RW00,Wi98}. On the other hand, the active systems allow computational entities to act and control their movements in order to accomplish a given task, which is our primary focus in this work. The most popular examples of active systems include metamorphic systems \cite{DSY04b,NGY00,WWA04}, swarm/mobile robotics \cite{CDP20,FPS12,PARK17,RCN14,SY10}, modular self-reconfigurable robotics \cite{ABD13,Fuk88,YSS07} and recent research on programmable matter \cite{DGPR16,DGRSS16}. Moreover, those robotic systems have received an increasing attention from the the engineering research community, and hence many solutions and frameworks have been produced for milli/micro-scale \cite{BG15,GKR10,KCL12} down to nanoscale systems \cite{DDL09,Ro06}.  

Shape transformations (sometimes called \emph{pattern formation}) can be seen as one of the most essential goals for almost every system among the vast variety of robotic systems including programmable matter and swarm robotic systems. In this work, we focus on a system of a two-dimensional square grid containing a collection of entities typically connected to each other and forming an initial connected shape $S_I$.
Each entity is equipped with a linear-strength mechanism that can push an entire line of consecutive entities one position in a single time-step in a given direction of a grid. The goal is to design an algorithm that can transform an initial shape $S_I$ into a given target shape $S_F$ through a chain of permissible moves and without losing the connectivity. That is, in each intermediate configuration we always want to guarantee that the graphs induced by the nodes occupied by the entities are connected. The connectivity-preservation is an important  assumption for  many practical applications, which usually require energy for data exchange as well as the implementation of various locomotion mechanisms.

\subsection{Related Work}
Many models of centralised or distributed coordination have been studied in the context of shape transformation problems. The assumed mechanisms  in those models can significantly influence the efficiency and feasibility of shape transformations. For example, the authors of \cite{AKT21,DP04,DSY04a,DSY04b,MSS19} consider mechanisms called sliding and rotation by which an agent can move and turn over neighbours through empty space. Under these models of individual movements, Dumitrescu and Pach \cite{DP04} and Michail \emph{et al.} \cite{MSS19} present universal transformations for any pair of connected shapes $(S_I,S_F)$ of the same size to each other. By restricting to rotation only, the authors in \cite{MSS19} proved that the decision problem of transformability is in $\mathbf{P}$; however, with a constant number of extra seed nodes connectivity preserving transformation can be completed with
 $\Omega(n^2)$  moves  \cite{MSS19}.

The alternative less costly reconfiguration solutions can be designed  by employing some parallelism,where multiple movements can occur at the same time, see theoretical studies in \cite{DDG18,DiLuna2019} and more practical implementation in \cite{RCN14}. Moreover, it has been shown that there exists a universal transformation with rotation and sliding that converts any pair of connected shapes to each other within $O(n)$ parallel moves in the worst case \cite{MSS19}. Also fast reconfiguration might be achieved by exploiting actuation mechanisms, where a single agent is now equipped with more strength to move many entities in parallel in a single time-step. A prominent example is the linear-strength model of Aloupis \emph{et al.} \cite{ABD13,ACD08}, where an entity is equipped with arms giving it the ability to extend/extract a neighbour, a set of individuals or the whole configuration in a single operation. Another elegant approach by Woods \emph{et al.} \cite{WCG13} studied another linear-strength mechanism by which an entity can drag a chain of entities parallel to one of the axes directions.

A more recent study along this direction is shown in \cite{AMP20}, and  introduces the \textit{line-pushing} model. In this model, an individual entity can push the whole line of consecutive entities one position in a given direction in a single time-step. As we shall explain, this model generalises some existing constant-strength models with a special focus on exploiting its parallel power for fast and more general transformations. Apart from the purely theoretical benefit of exploring fast reconfigurations, this model also provides a practical framework for more efficient reconfigurations in real systems. For example, self-organising robots could be reconfiguring into multiple shapes in order to pass through canals, bridges or corridors in a mine. In another domain, individual robots could be containers equipped with motors that can push an entire row to manage space in large warehouses. Another future application could be a system of very tiny particles injected into a human body and transforming into several shapes in order to efficiently traverse through the veins and capillaries and treat infected cells. 

This model is capable of simulating some constant-strength models. For example, it can simulate the sliding and rotation model \cite{DP04,MSS19} with an increase in the worst-case running time only by a factor of 2. This implies that all universality and reversibility properties of individual-move transformations still hold true in this model. Also, the model allows the diagonal connections on the grid. Several sub-quadratic time centralised transformations have been proposed, including an $O(n \sqrt{n} )$-time universal transformation that preserves the connectivity of the shape during its course \cite{AMP2020}. By allowing transformations to disconnect the shape during their course, there exists a centralised universal transformation that completes within $O(n \log n )$ time.

Another recent related set of models studied in \cite{CDP20,FGHKSS21,GHK19} consider a single robot which moves over a static shape consisting of tiles and the goal is for the robot to transform the shape by carrying one tile at a time. In those systems, the single robot which controls and carries out the transformation is typically modelled as a finite automaton. Those models can be viewed as partially centralised as on one hand they have a unique controller but on the other hand that controller is operating locally and suffering from a lack of global information.

\subsection{Our Contribution}
\label{sec:Contribution}

In this work, our main objective is to give the first distributed transformations for programmable matter systems implementing the linear-strength mechanism of the model of line moves. All existing transformations for this model are centralised, thus, even though they reveal the underlying transformation complexities, they are not directly applicable to real programmable matter systems. Our goal is to develop distributed transformations that, if possible, will preserve all the good properties of the corresponding centralised solutions. These include the \emph{move complexity} (i.e., the total number of line moves) of the transformations and their ability to preserve the connectivity of the shape throughout their course.

However, there are considerable technical challenges that one must deal with in order to develop such a distributed solution. As will become evident, the lack of global knowledge of the individual entities and the condition of preserving connectivity greatly complicate the transformation, even when restricted to special families of shapes. Timing is an essential issue as the line needs to know when to start/stop pushing. When moving or turning, all agents of the line must follow the same route, ensuring that no one is being pushed off. There is an additional difficulty due to the fact that agents do not automatically know whether they have been pushed (but it might be possible to infer this through communication and/or local observation).

Consider a discrete system of $n$ simple indistinguishable devices, called \emph{agents}, forming a connected shape $S_I$ on a two-dimensional square grid. Agents act as finite-state automata (i.e., they have constant memory) that can observe the states of nearby agents in a Moore neighbourhood (i.e., the eight cells surrounding an agent on the square gird).  They operate in synchronised Look-Compute-Move (LCM) cycles on the grid. All communication is local, and actuation is based on this local information as well as the agent's internal state.

Let us consider a very simple distributed transformation of a diagonal line shape $S_D$ into a straight line $S_L$, $|S_D| = |S_L| = n$, in which all agents execute the same procedure in parallel synchronous rounds. In general, the diagonal appears to be a hard instance because any parallelism related to line moves that might potentially be exploited does not come for free. Initially, all agents are occupying the consecutive diagonal cells on the grid  $(x_1,y_1), (x_1+1, y_1+1), \ldots ,(x_1+n-1, y_1+n-1)$. In each round, an agent $p_i = (x,y)$ moves one step down if $(x-1,y-1)$ is occupied, otherwise it stays still in its current cell. After $O(n)$ rounds, all agents form $S_L$ within a total number of $1 + 2 + \ldots +n = O(n^2)$ moves, while preserving connectivity during the transformation (throughout, connectivity includes horizontal, vertical, and diagonal adjacency). See Figure \ref{fig:Simulation}.

\begin{figure}[hbt!] 
	\centering 
	\includegraphics[scale=0.6]{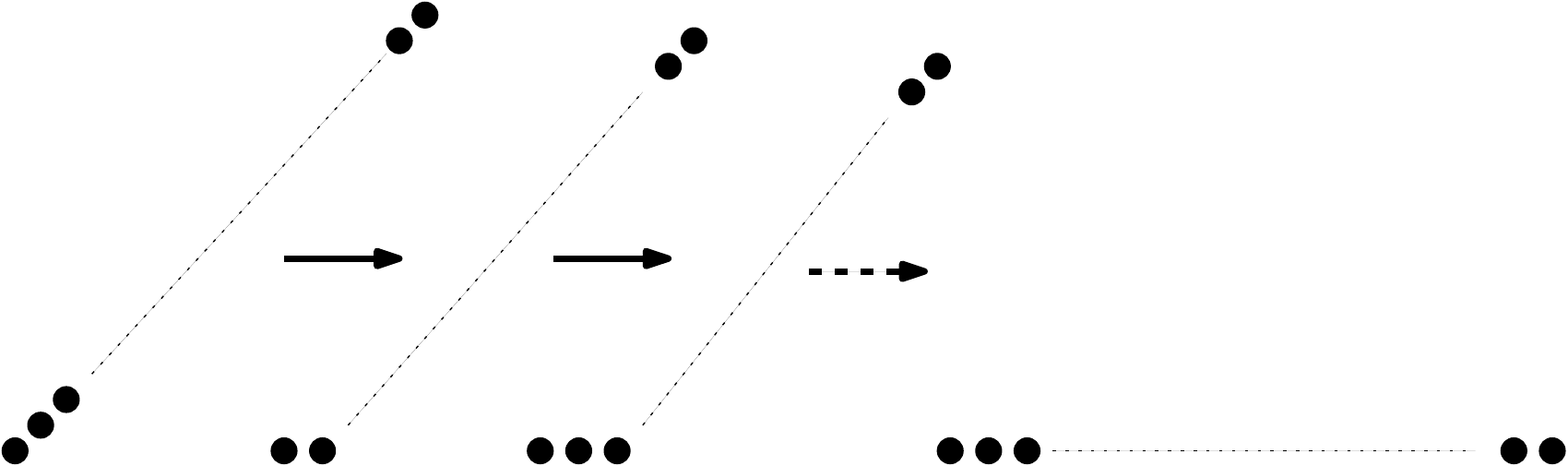}
	\caption{A simulation of the simple procedure. From left to right, rounds $0,1,2,\ldots,n$.}
	\label{fig:Simulation}
\end{figure}

The above transformation, even though time-optimal has a move complexity asymptotically equal to the worst-case single-move distance between $S_I$ and $S_F$. This is because it always moves individual agents, thus not exploiting the inherent parallelism of line moves. Our goal, is to trade time for number of line moves in order to develop alternative distributed transformations which will complete within a sub-quadratic number of moves. 
Given that actuation is a major source of energy consumption in real programmable matter and robotic systems, moves minimisation is expected to contribute in the deployment and implementation of energy-efficient systems.

We already know that there is a centralised $O(n \log n)$-move connectivity-preserving transformation, working for a large family of connected shapes \cite{AMP2020}. That centralised strategy transforms a pair of connected shapes \textbf{$(S_I,S_F)$} of the same order (i.e., the number of agents) to each other, when the associated graphs of both shapes contain a Hamiltonian path (see also Itai \emph{et al.}  \cite{IPS82} for rectilinear Hamiltonian paths), while preserving connectivity during the transformation. This approach initially forms a line from one endpoint of the Hamiltonian path, then flattens all agents along the path gradually via line moves, while successively doubling the line length in each round. After $O(n \log n)$ moves, it arrives at the final straight line $S_L$ of length $n$, which can be then transformed into $S_F$ by reversing the transformation of $S_F$ into $S_L$, within the same asymptotic number of moves. 

In this work, we introduce the first distributed transformation exploiting the linear-strength mechanism of the \textit{line-pushing} model. It provides a solution to  the line formation problem, that is, for any initial Hamiltonian shape $S_I$, form a final straight line $S_L$ of the same order. It is essentially a distributed implementation of the centralised Hamiltonian transformation of \cite{AMP2020}. We show that it preserves the asymptotic bound of $O(n \log n)$ line moves (which is still the best-known centralised bound), while keeping the whole shape connected throughout its course. This is the first step towards distributed transformations between any pair of Hamiltonian shapes. The inverse of this transformation ($S_L$ into $S_I$) appears to be a much more complicated problem to solve as the agents need to somehow know an encoding of the shape to be constructed and that in contrast to the centralised case, reversibility does not apply in a straightforward way. Hence, the reverse of this transformation ($S_L$ into $S_I$) is left as a future research direction.

We restrict attention to the class of Hamiltonian shapes. This class, apart from being a reasonable first step in the direction of distributed transformations in the given setting, might give insight to the future development of universal distributed transformations, i.e., distributed transformations working for any possible pair of initial and target shapes. This is because geometric shapes tend to have long simple paths. For example, the length of their longest path is provably at least $\sqrt{n}$. We here focus on developing efficient distributed transformations for the extreme case in which the longest path is a Hamiltonian path. However, one might be able to apply our Hamiltonian transformation to any pair of shapes, by, for example, running a different or similar transformation along branches of the longest path and then running our transformation on the longest path. We leave how to exploit the longest path in the general case (i.e., when initial and target shapes are not necessarily Hamiltonian) as an interesting open problem.

We assume that a pre-processing phase provides the Hamiltonian path, i.e., a global sense of direction is made available to the agents through a labelling of their local ports (e.g., each agent maintains two local ports incident to its predecessor and successor on the path). Similar assumptions exist in the literature of systems of complex shapes that contain a vast number of self-organising and limited entities. A prominent example is \cite{RCN14} in which the transformation relies on an initial central phase to gain some information about the number of entities in the system. 

Now, we are ready to sketch a high-level description of the transformation. A Hamiltonian path $P$ in the initial shape $S_I$ starts with a head on one endpoint labelled $l_h$, which is leading the process and coordinating all the sub-procedures during the transformation. The transformation proceeds in $\log n$ phases, each consisting of six sub-phases (or sub-routines) and every sub-phase running for one or more synchronous rounds. Figure \ref{fig:I_Phase} gives an illustration of a phase of this transformation when applied on the diagonal line shape. Initially, the head $l_h$ forms a trivial line of length 1. By the beginning of each phase $i$, $0 \le i \le \log n - 1$, there exists a line $L_i$ starting from the head $l_h$ and ending at a tail $l_t$ with $2^i-2$ internal agents labelled $l$ in between. By the end of phase $i$, $L_i$ will have doubled its length as follows. 

\begin{figure}[hbt!]
		\centering
		\subfigure[\textsf{DefineSeg}, \textsf{CheckSeg} and \textsf{DrawMap}.]{\includegraphics[scale=0.45]{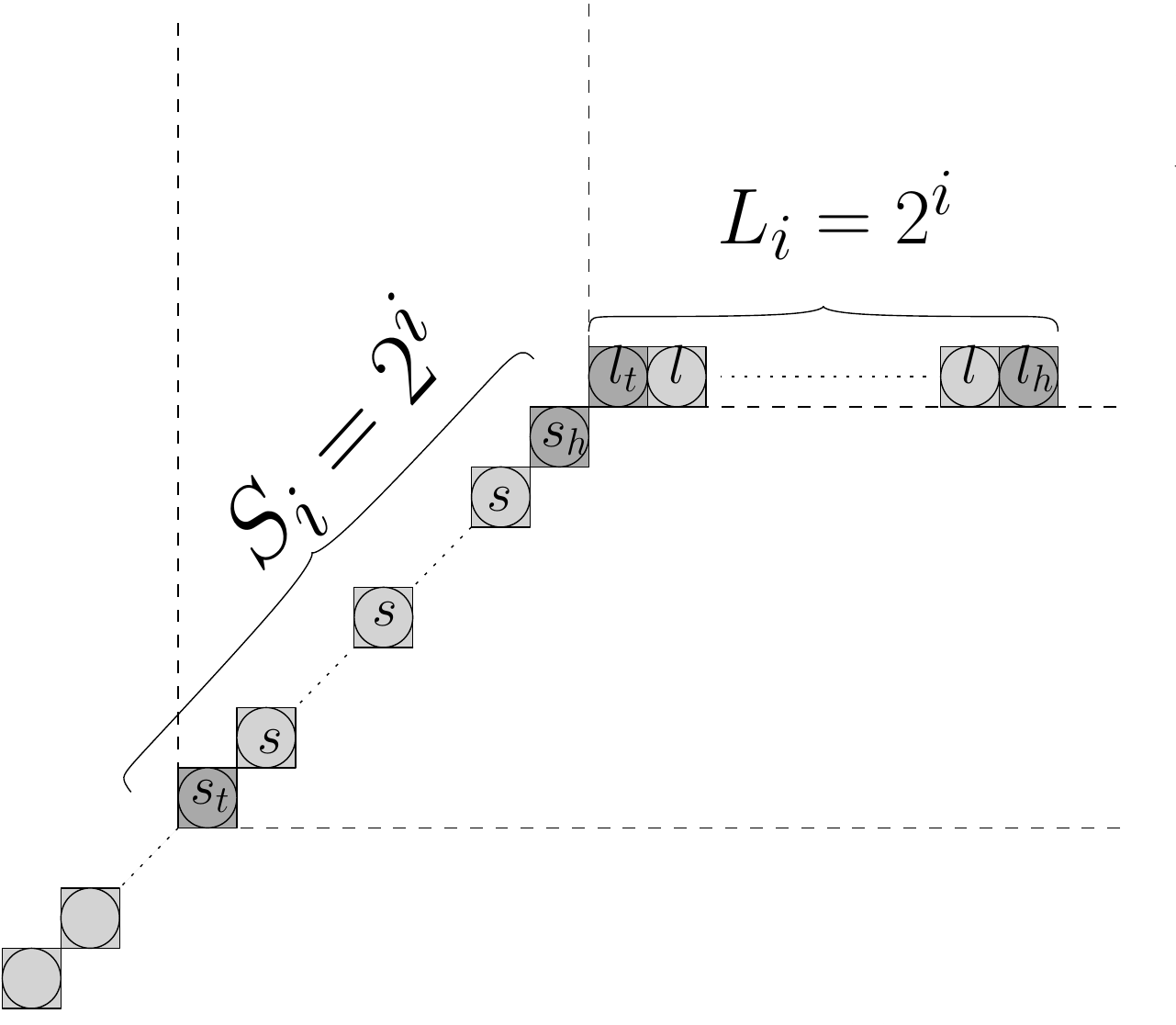}} \quad
		\subfigure[\textsf{Push}.]{\includegraphics[scale=0.45]{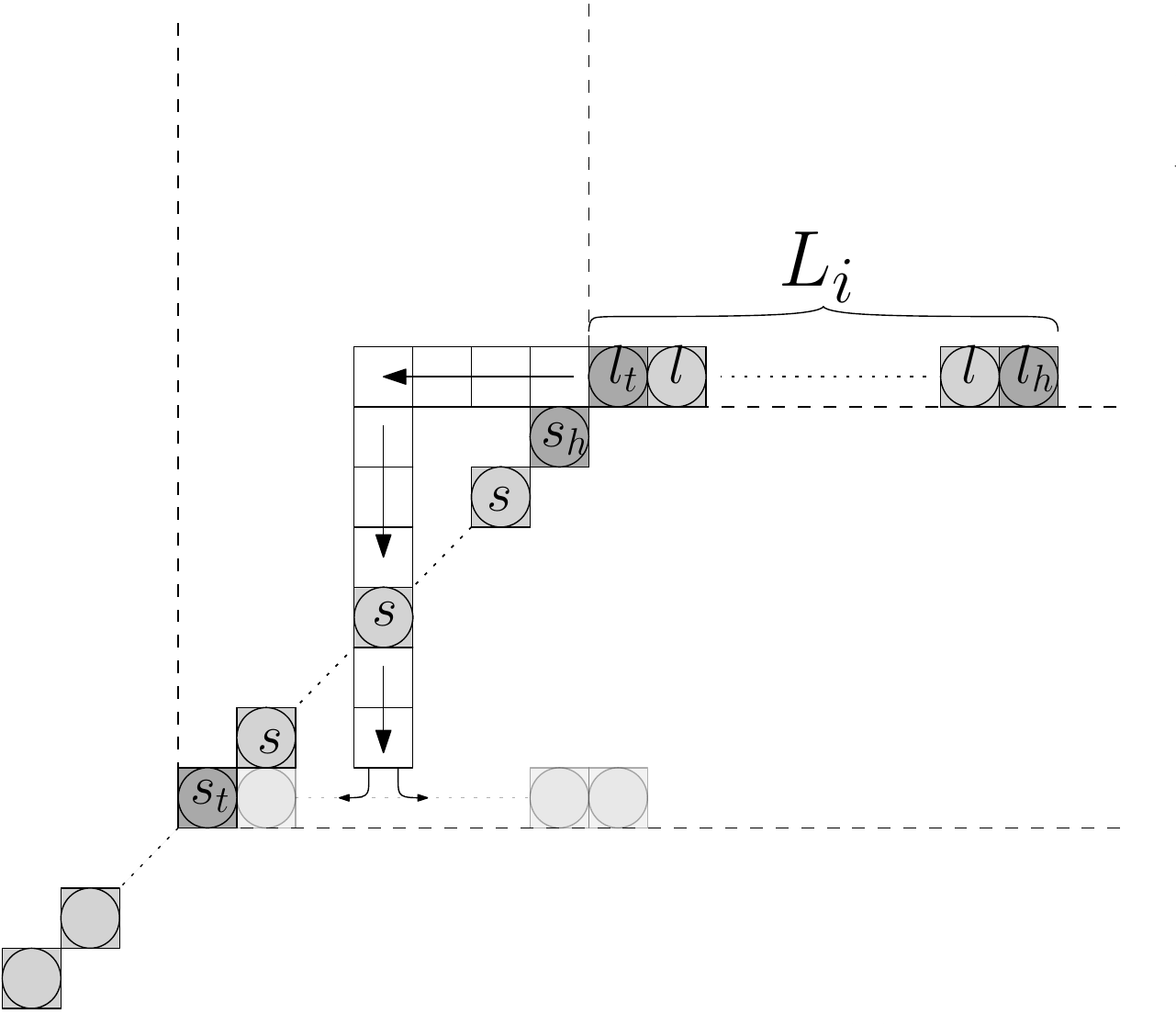}} \quad
		\subfigure[\textsf{RecursiveCall}.]{\includegraphics[scale=0.45]{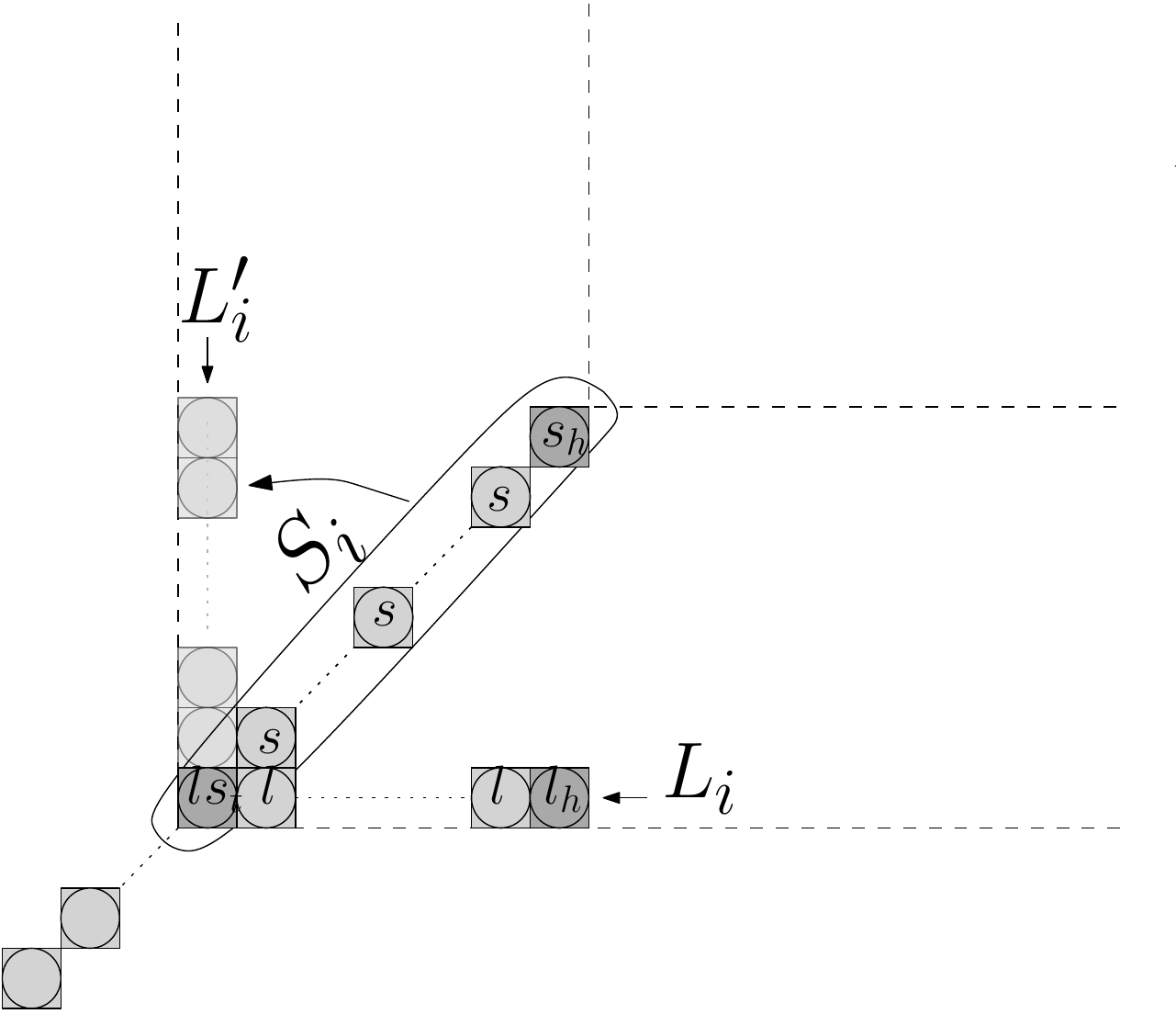}}\quad
		\subfigure[\textsf{Merge}.]{\includegraphics[scale=0.45]{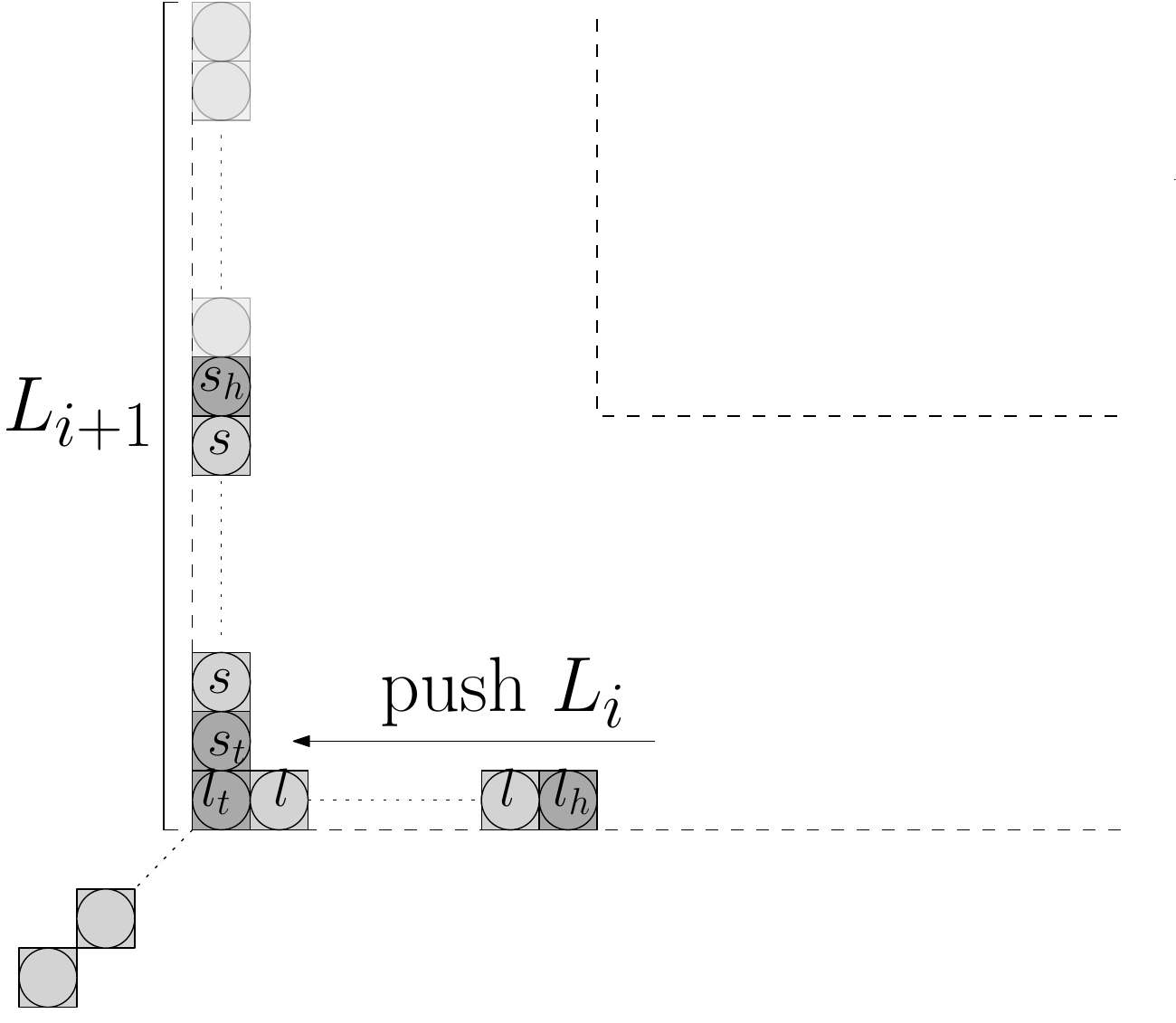}}\\
		
		\caption{From \cite{AMP2020}, a snapshot of phase $i$ of the Hamiltonian transformation on the shape of a diagonal line. Each occupied cell shows the current label state of an agent. Light grey cells show ending cells of the corresponding moves.}
		\label{fig:I_Phase}
\end{figure} 

In \textsf{DrawMap}, $l_h$ designates a route on the grid through which $L_i$ pushes itself towards the tail $s_t$ of $S_i$. It consists of two primitives: \textsf{ComputeDistance} and \textsf{CollectArrows}. In  \textsf{ComputeDistance}, the line agents act as a distributed counter to compute the Manhattan distance between the tails of $L_i$ and $S_i$. In \textsf{CollectArrows}, the local directions are gathered from $S_i$'s agents and distributed into $L_i$'s agents, which collectively draw the route map. Once this is done, $L_i$ becomes ready to move and $l_h$ can start the \textsf{Push} sub-phase. During pushing, $l_h$ and $l_t$ synchronise the movements of $L_i$'s agents as follows: (1) $l_h$ pushes while $l_t$ is guiding the other line agents through the computed route and (2) both are coordinating any required swapping of states with agents that are not part of $L_i$ but reside in $L_i$'s trajectory. Once $L_i$ has traversed the route completely, $l_h$ calls \textsf{RecursiveCall} to apply the general procedure recursively on $S_i$ in order to transform it into a line $L_i^{\prime}$. Figure \ref{fig:DignalPattern} shows a graphical illustration of the core recursion on the special case of a diagonal line shape. Finally, the agents of $L_i$ and $L_i^{\prime}$ combine into a new straight line $L_{i+1}$ of $2^{i+1}$ agents through the \textsf{Merge} sub-procedure. Then, the head $l_h$ of $L_{i+1}$ begins a new phase $i+1$. 

\begin{figure}[hbt!]
	\centering
	\includegraphics[scale=0.65]{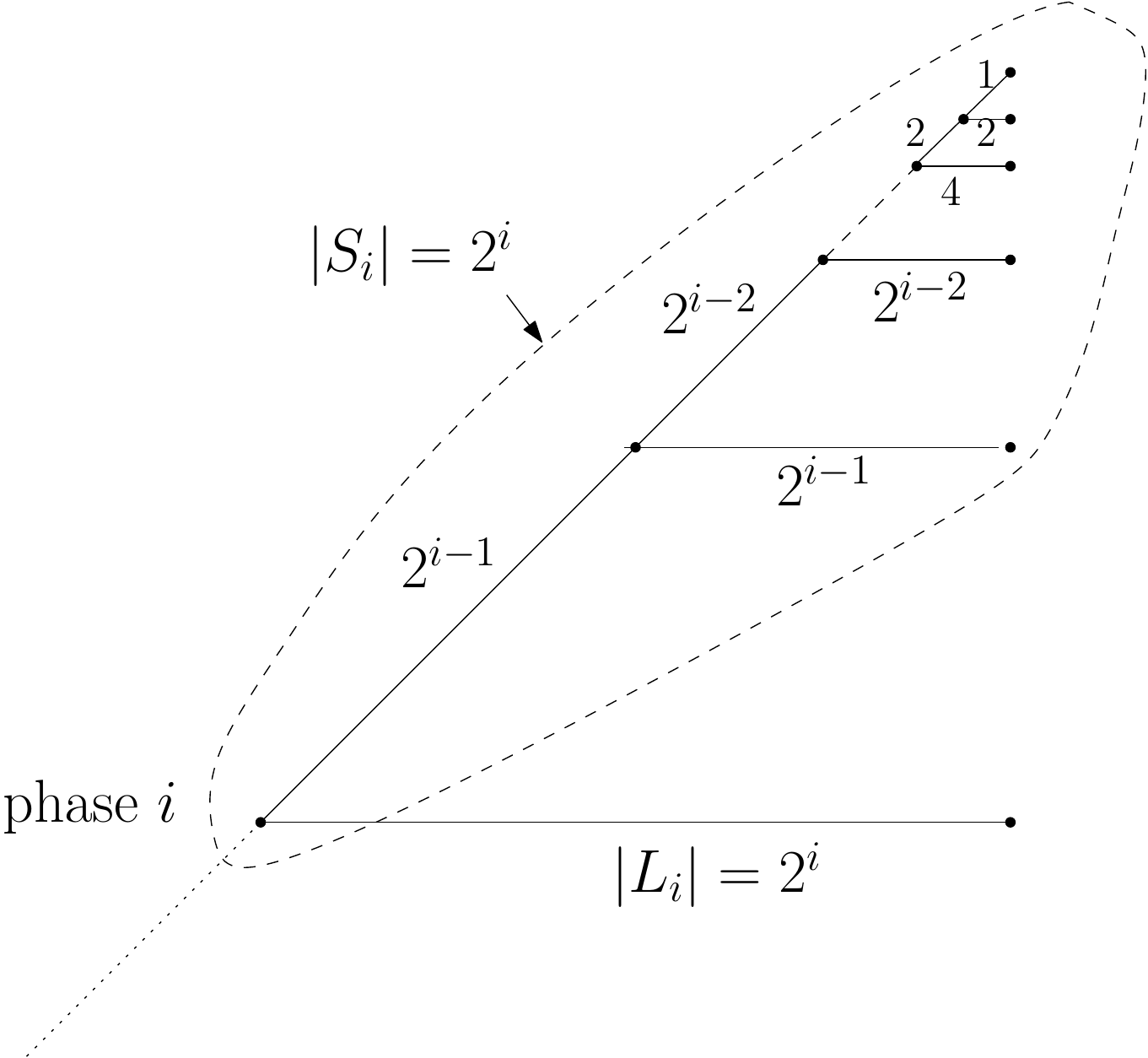}
	\caption{A zoomed-in picture of the core recursive technique \textsf{RecursiveCall} in Figure \ref{fig:I_Phase}(c).}
	\label{fig:DignalPattern}
\end{figure}

Section \ref{sec:model} formally defines the model and the problem under consideration. Section \ref{sec:DistributedHamiltonianTransformation} presents our distributed connectivity-preserving transformation that solves the line formation problem for Hamiltonian shapes, achieving a total of $O(n \log n)$ line moves.  

\section{Model}
\label{sec:model}
We consider a system consisting of $n$ agents forming a connected shape $S$ on a two-dimensional square grid in which each agent $p \in S$ occupies a unique cell $cell(p) =(x,y)$, where $x$ indicates columns and $y$ represents rows. Throughout, an agent shall also be referred to by its coordinates. Each cell $(x,y)$ is surrounded by eight adjacent cells in each cardinal and ordinal direction, (\emph{N}, \emph{E}, \emph{S}, \emph{W}, \emph{NE}, \emph{NW}, \emph{SE}, \emph{SW}). At any time, a cell $(x,y)$ can be in one of two states, either empty or occupied. An agent $p \in S $ is a \emph{neighbour} of (or \emph{adjacent} to) another agent $p^{\prime} \in S$, if $p^{\prime} $ occupies one of the eight adjacent cells surrounding $p$, that is their coordinates satisfy $p^{\prime}_x -1 \le p_x \le p^{\prime}_x +1$ and $p^{\prime}_y -1 \le p_y \le p^{\prime}_y +1$, see Figure \ref{fig:Compass}. For any shape $S$, we associate a graph $G(S) = (V, E)$ defined as follows, where $V$ represents agents of $S$ and $E$ contains all pairs of adjacent neighbours, i.e. $(p, p^{\prime}) \in E$ iff $p$ and $p^{\prime}$ are neighbours in $S$. We say that a shape $S$ is connected iff $G(S)$ is a connected graph. The \emph{distance} between agents $p \in S$ and $p^{\prime} \in S$ is defined as the Manhattan distance between their cells, $\Delta(p, p^{\prime}) = |p_x - p^{\prime}_x|+|p_y - p^{\prime}_y|$. A shape $S$ is called \emph{Hamiltonian shape} iff $G(S)$ contains a Hamiltonian path, i.e. a path starting from some $p \in S$, visiting every agent in $S$ and ending at some $p^{\prime} \in S$, where $p \ne p^{\prime}$.

\begin{figure}[hbt!] 
	\centering 
	\includegraphics[scale=0.5]{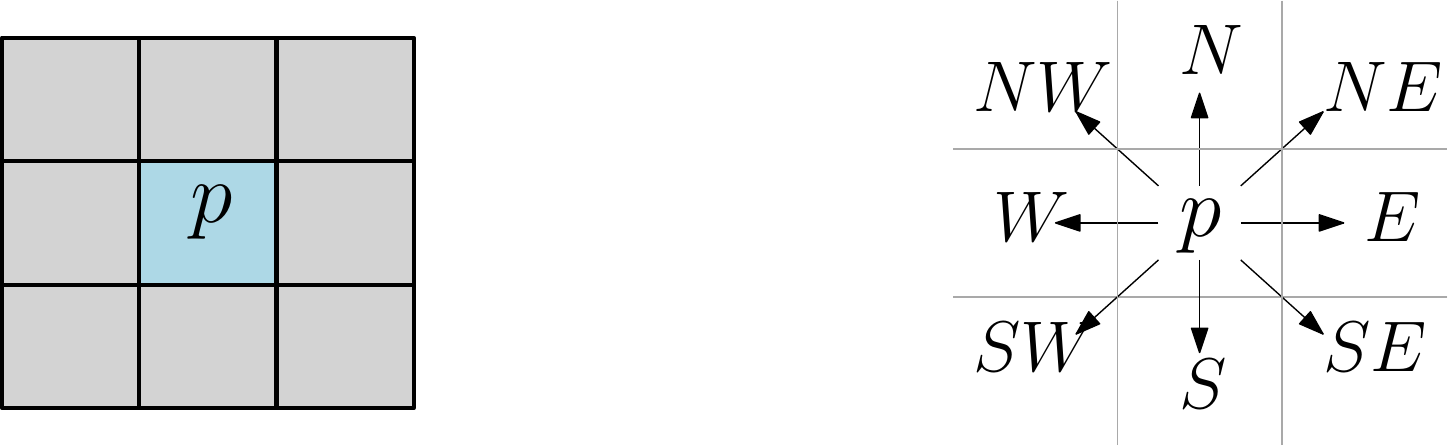}
	\caption{An agent $p$ is a neighbour to any agent locating at one of the eight surrounding cells in grey.}
	\label{fig:Compass}
\end{figure}

In this work, each agent is equipped with the linear-strength mechanism introduced in \cite{AMP20}, called the \emph{line pushing mechanism}. A line $L$ consists of a sequence of $k$ agents occupying consecutive cells on the grid, say w.l.o.g, $L= (x,y), (x+1, y), \ldots, (x+k-1, y)$, where $1 \le k \le n$. The agent $p \in L$ occupying $(x,y)$ is capable of performing an operation of a \textbf{line move} by which it can push all agents of $L$ one position rightwards to positions $(x+1,y), (x+2, y), \ldots, (x+k, y)$ in a single time-step. The \emph{line moves} towards the ``down'', ``left'' and ``up'' directions are defined symmetrically by rotating the system $ 90^{\circ}$, $180^{\circ} $ and $ 270^{\circ} $ clockwise, respectively. From now on, this operation may be referred to as \emph{move}, \emph{movement} or \emph{step}. We call the number of agents in $S$ the \emph{size} or \emph{order} of the shape, and throughout this work all logarithms are to the base 2.

We assume that the agents share a sense of orientation through a consistent labelling of their local ports. Agents do not know the size of $S$ in advance neither they have any other knowledge about $S$. Each agent has a constant memory (of size independent of $n$) and a local visibility mechanism by which it observes the states of its eight neighbouring cells simultaneously. The agents act as finite automata operating in synchronous rounds consisting of LCM steps. Thus, in every discrete round, an agent observes its own state and for each of its eight adjacent cells, checks whether it is occupied or not. For each of those occupied, it also observes the state of the agent occupying that cell. Then, the agent updates its state or leaves it unchanged and performs a \emph{line move} in one direction \emph{d} $\in \{$\emph{up, down, right, left}$\}$ or stays still. A \emph{configuration} $C$ of the system is a mapping from $\mathbb{Z}_{\geq 0}^2$ to $\{0\} \cup Q$, where $Q$ is the state space of agents. We define $S(C)$ as the shape of configuration $C$, i.e., the set of coordinates of the cells occupied in $S$. Given a configuration $C$, the LCM steps performed by all agents in the given round, yield a new configuration $C^{\prime}$ and the next round begins. If at least one move was performed, then we say that this round has transformed $S(C)$ to $S(C^{\prime})$.

Throughout this work, we assume that the initial shape $S_I$ is Hamiltonian and the final shape is a straight line $S_L$, where both $S_I$ and $S_L$ have the same order. We also assume that a pre-elected leader is provided at one endpoint of the  Hamiltonian path of $S_I$. It is made available to the agents in the distributed way that each agent $p_i$ knows the local port leading to its predecessor $p_{i-1}$ and its successor $p_{i+1}$, for all $1 \le i \le n$.

An agent $p \in S$ is defined as a 5-tuple $(X,M,Q,\delta,O)$, where
$Q$ is a finite set of states, $X$ is the input alphabet representing the states of the eight cells that surround an agent $p$ on the square grid, so $|X|={|Q|}^8$, $M = \{\uparrow, \downarrow, \rightarrow, \leftarrow, none\}$ is the output alphabet
corresponding to the set of moves, a transition function $ \delta: Q \times X \rightarrow Q \times  M$ and the output function $O: \delta \times X \rightarrow M$.

\subsection{Problem definition} 

We now formally define the problem considered in this work.\\

\noindent\textbf{{\sc HamiltonianLine}.} Given any initial Hamiltonian shape $S_I$, the agents must form a final straight line $S_L$ of the same order from $S_I$ via line moves while preserving connectivity throughout the transformation.

\section{The Distributed Hamiltonian Transformation}
\label{sec:DistributedHamiltonianTransformation}

In this section, we develop a distributed algorithm exploiting line moves to form a straight line $S_L$ from an initial connected shape $S_I$ which is associated to a graph that contains a Hamiltonian path. As we will argue, this strategy performs $O(n \log n)$ moves, i.e., it is as efficient w.r.t. moves as the best-known centralised transformation \cite{AMP2020}, and completes within $O(n^2 \log n)$ rounds, while keeping the whole shape connected during its course.

We assume that through some pre-processing the Hamiltonian path $P$ of the initial shape $S_I$ has been made available to the $n$ agents in a distributed way. $P$ starts and ends at two agents, called the head $p_1$ and the tail $p_n$, respectively. The head $p_1$ is leading the process (as it can be used as a pre-elected unique leader) and is responsible for coordinating and initiating all procedures of this transformation. In order to simplify the exposition, we assume that $n$ is a power of 2; this can be easily dropped later. The transformation proceeds in $\log n$ phases, each of which consists of six sub-phases (or sub-routines). Every sub-phase consist of one or more synchronous rounds. The transformation starts with a trivial line of length 1 at the head's endpoint, then it gradually flattens all agents along $P$ gradually while successively doubling its length, until arriving at the final straight line $S_L$ of length $n$. 

A state $q \in Q$ of an agent $p$ will be represented by a vector with seven components $(c_1, c_2, c_3, c_4, c_5,$ $ c_6, c_7)$. The first component $c_1$ contains a label $\lambda$ of the agent from a finite set of labels $\Lambda$, $c_2$ is the transmission state that holds a string of length at most three, where each symbol of the string can either be a special mark $w$ from a finite set of marks $W$ or an arrow direction $a \in A = \{\rightarrow, \leftarrow, \downarrow, \uparrow, \nwarrow, \nearrow, \swarrow,\searrow\}$ and $c_3$ will store a symbol from $c_2$'s string, i.e., a special mark or an arrow. The local Hamiltonian direction $a \in A$ of an agent $p$ indicating  predecessor and successor is recorded in $c_4$, the counter state $c_5$ holds a bit from $\{0,1\}$, $c_6$ stores an arrow $a \in A$ for map drawing (as will be explained later) and finally $c_7$ is holding a pushing direction $d \in M$. The ``$\cdot$'' mark indicates an empty component; a non-empty component is always denoted by its state. An agent $p$ may be referred to by its label $\lambda \in \Lambda$ (i.e., by the state of its $c_1$ component) whenever clear from context.  

By the beginning of phase $i$, $0 \le i \le \log n - 1$, there exists a terminal straight line $ L_i $ of $2^i$ active agents occupying a single row or column on the grid, starting with a head labelled $l_h$ and ending at a tail labelled $ l_t $, while internal agents have label $l$. All agents in the rest of the configuration are inactive and labelled $k$. During phase $i$, the head $l_h$ leads the execution of six sub-phases:

\begin{enumerate}
	\item \textsf{DefineSeg}: Identify the next segment $S_i$ of length $2^i$ in the Hamiltonian path.
	\item \textsf{CheckSeg}: Check whether $S_i$ is in line or perpendicular line to $L_i$. Go to (6) if perpendicular or start phase $i+1$ otherwise.  
	\item \textsf{DrawMap}: Compute a rout map that takes $ L_i $ to the end of $S_i$.
	\item \textsf{Push}: Move $ L_i $ along the drawn route map.
	\item \textsf{RecursiveCall}: A recursive-call on $S_i$ to transform it into a straight line $L_i^{\prime }$.
	\item \textsf{Merge}: Combine $L_i$ and $L_i^{\prime }$ together into a straight line $L_{i+1}$ of $2^{i+1}$ double length. Then,  phase $i+1$ begins.  
\end{enumerate}

Figure \ref{fig:I_Phase} gives an illustration of a phase of this transformation when applied on the diagonal line shape. First, it identifies the next $2^i$ agents on $P$. These agents are forming a segment $S_i$ which can be in any configuration. To do that, the head emits a signal which is then forwarded by the agents along the line. Once the signal arrives at $S_i$, it will be used to re-label $S_i$ so that it starts from a head in state $s_h$, has $2^i-2$ internal agents in state $s$, and ends at a tail $s_t$; this completes the \textsf{DefineSeg} sub-phase. Then, $l_h$ calls \textsf{CheckSeg} in order to check whether the line defined by $S_i$ is in line or perpendicular to $L_i$. This can be easily achieved through a moving state initiated at $L_i$ and checking for each agent of $S_i$ its local directions relative to its neighbours. If the check returns true, then $l_h$ starts a new round $i+1$ and calls \textsf{Merge} to combine $L_i$ and $S_i$ into a new line $L_{i+1}$ of length $2^{i+1}$. Otherwise, $l_h$ proceeds with the next sub-phase, \textsf{DrawMap}.

In \textsf{DrawMap}, $l_h$ designates a route on the grid through which $L_i$ pushes itself towards the tail $s_t$ of $S_i$. It consists of two primitives: \textsf{ComputeDistance} and \textsf{CollectArrows}. In  \textsf{ComputeDistance}, the line agents act as a distributed counter to compute the Manhattan distance between the tails of $L_i$ and $S_i$. In \textsf{CollectArrows}, the local directions are gathered from $S_i$'s agents and distributed into $L_i$'s agents, which collectively draw the route map. Once this is done, $L_i$ becomes ready to move and $l_h$ can start the \textsf{Push} sub-phase. During pushing, $l_h$ and $l_t$ synchronise the movements of $L_i$'s agents as follows: (1) $l_h$ pushes while $l_t$ is guiding the other line agents through the computed route and (2) both are coordinating any required swapping of states with agents that are not part of $L_i$ but reside in $L_i$'s trajectory. Once $L_i$ has traversed the route completely, $l_h$ calls \textsf{RecursiveCall} to apply the general procedure recursively on $S_i$ in order to transform it into a line $L_i^{\prime}$. Figure \ref{fig:DignalPattern} shows a graphical illustration of the core recursion on the special case of a diagonal line shape. Finally, the agents of $L_i$ and $L_i^{\prime}$ combine into a new straight line $L_{i+1}$ of $2^{i+1}$ agents through the \textsf{Merge} sub-procedure. Then, the head $l_h$ of $L_{i+1}$ begins a new phase $i+1$. Now, we are ready to proceed with the detailed description of each sub-phase.

\subsection{Define the next segment $S_i$}
\label{sec:Define}
This sub-phase identifies the following segment $S_i$ and activates its $2^i$ agents, which instantly follow the Hamiltonian path's terminal straight line $L_i$ of length $2^i$. The algorithm works as follows: The line head $l_h$ transmits a special mark ``$\textcircled{\scriptsize H} $'' to go through all active agents in the Hamiltonian path $P$. It updates its transmission component $c_2$ as follows: $\delta(l_h,\cdot,\cdot,a\in A,\cdot,\cdot,\cdot)= (l_h,\textcircled{\scriptsize H},\cdot, a \in A,\cdot,\cdot,\cdot)$. This is propagated by active agents by always moving from a predecessor $p_{i}$ to a successor $p_{i+1}$, until it arrives at the first inactive agent with label $k$, which then becomes active and the head of its segment by updating its label as $\delta(k,\textcircled{\scriptsize H},\cdot,a \in A,\cdot,\cdot,\cdot)= (s_h,\cdot,\cdot, a \in A,\cdot,\cdot,\cdot)$. Similarly, once a line agent $p_i$ passes ``$\textcircled{\scriptsize H} $'' to $p_{i+1}$, it also initiates and propagates its own mark ``$\textcircled{\scriptsize l} $'' to activate a corresponding segment agent $s$. The line tail $l_t$ emits ``$\textcircled{\scriptsize T} $'' to activate the segment tail $s_t$, which in turn bounces off a special end mark ``$\otimes$'' announcing the end of \textsf{DefineSeg}. By that time, the next segment $S_i$ consisting of $2^i$ agents, starting from a head labelled $s_h$, ending at a tail $s_t$ and having $2^i-2$ internal agents with label $s$, has been defined. The ``$\otimes$'' mark is propagated back to the head $l_h$ along the active agents, by always moving from $p_{i+1}$ to $p_{i}$.

\begin{lemma} \label{lem:DefineNextSegCorrectness}
	\textsf{DefineSeg} correctly activates all agents of $S_i$ in $O(n)$ rounds. 
\end{lemma}  
\begin{proof}
	When an active agent $p_i$ with label inline $l$ or tail $l_t$ observes the head mark ``$\textcircled{\scriptsize H} $'' on the state of its predecessor $p_{i-1}$, it then updates transmission state $c_2$ to ``$\textcircled{\scriptsize H} $'' and  initiates a special mark on its waiting state $c_4$. This can be either inline ``$\textcircled{\scriptsize L} $'' or tail ``$\textcircled{\scriptsize T} $'' mark.  Once an inactive agent notices predecessor with ``$\textcircled{\scriptsize L} $'' or ``$\textcircled{\scriptsize T} $'' mark, it activates and changes its label $c_1$ to the corresponding state, ``$s$'' or ``$s_t$'', respectively. 	Immediately after activating the tail $s_t$, it bounces off a special end mark ``$\otimes$'' transmitted along all active agents back to the head $l_h$ of the line to indicate the end of this sub-phase.  That is, the tail $s_t$ sets ``$\otimes$'' in transmission state, so when agent $p_i$ observes successor $p_{i+1}$ showing ``$\otimes$'', it updates its transmission state to $c_2 \gets \otimes$. When witnessing predecessor or successor with an empty transmission state, each agent resets its $c_2$ state to ``$\cdot$''.  Once the head $l_t$ detects the ``$\otimes$'' mark, it then calls the next sub-routine, \textsf{CheckSeg}. Because the transformation always doubles the length of the straight line, the line $L_i$ cannot be of odd length, unless it is originally composed of 1 agent labelled $l_h$ and adjacent to an inactive neighbour on the path $P$. In this case, the adjacent agent activates when it observes the head mark, updates label to $s_h$ and reflects an end special mark ``$\otimes$'' back to $l_h$. \qed
\end{proof}

Figure \ref{fig:DNS1} depicts an implementation of \textsf{DefineSeg} on a straight line of four agents, in which the next segment $S_i$ is represented as a line for clarity, but it can be of any configuration. All transitions of this sub-routine is given in Algorithm \ref{alg:DefNxSg}, excluding all that have no effect.

\begin{figure}[hbt!]
	\centering 
	\includegraphics[scale=0.72]{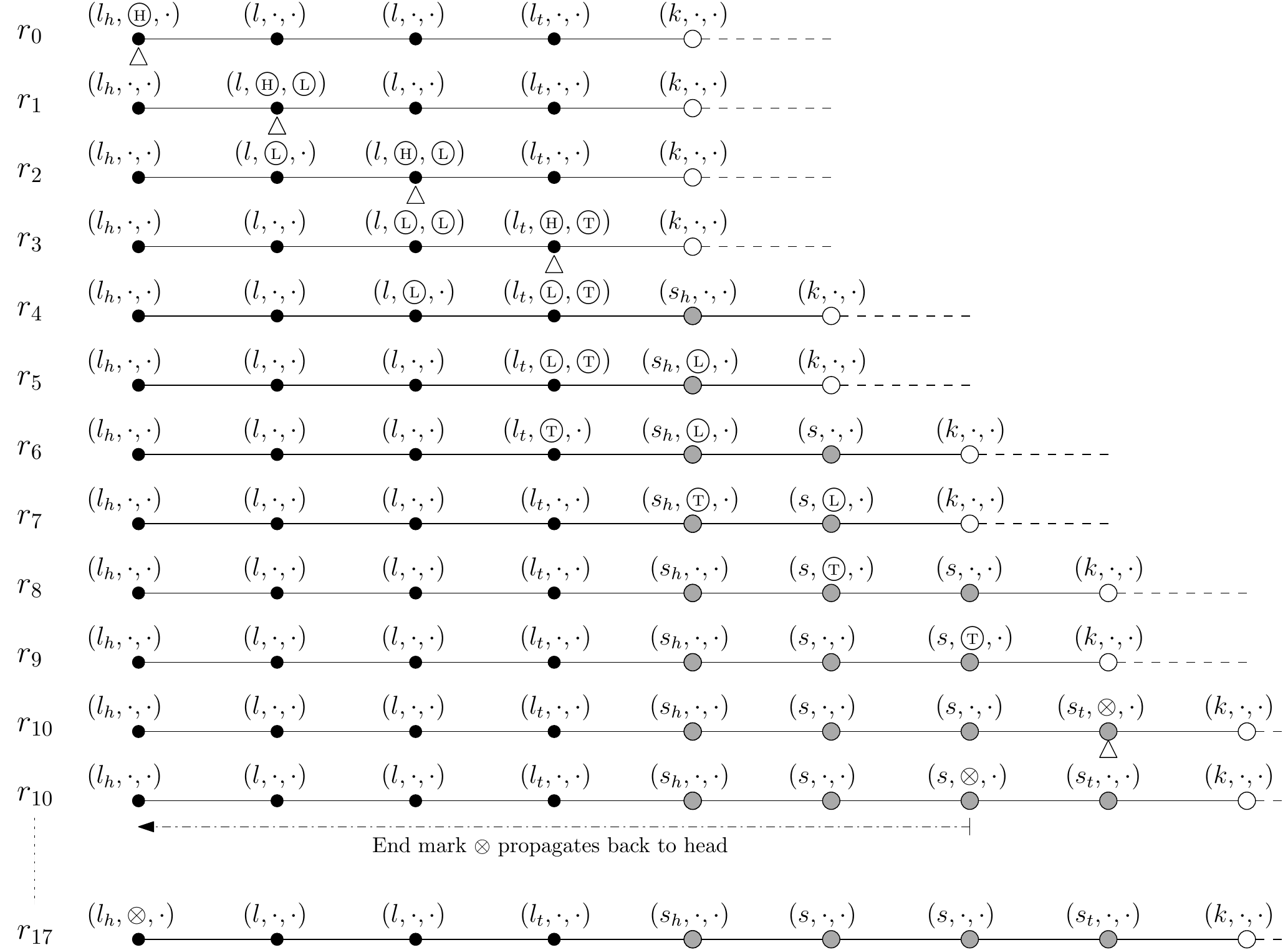}
	\caption[An implantation of \textsf{DefineSeg} on a line $L_i$ of four agents.]{An implantation of \textsf{DefineSeg} on a line $L_i$ of four agents depicted as black dots. Each agent uses only its 3 state components $(c_1,c_2,c_3)$, where $c_1$ is the label state, $c_2$ the transmission state and $c_3$ the waiting mark state. In round $r=0$,  $L_i$ is labelled correctly, starting from a head $l_h$ and ends at a tail $l_t$ with internal agents $l$. Inactive agents with circles are labelled $k$. First, $l_h$ sets $c_2 \gets\textcircled{\scriptsize H}$. Thereafter, when an active agent $p_i$ notices successor $p_{i-1}$ showing ``$\textcircled{\scriptsize H}$'' , it updates to $c_2 \gets\textcircled{\scriptsize H}$ (a small triangle indicates this initialisation in rounds $r_0, r_1,r_2, r_3,$ and $r_{10}$). Agents of label $l$ and $l_t$ propagate ``$\textcircled{\scriptsize L}$'' and ``$\textcircled{\scriptsize T}$'', respectively. Whenever an inactive agent sees predecessor presenting a mark, it activates (grey dots) and updates label to corresponding state. Once activating the segment tail $s_t$, it propagates an end mark ``$\otimes$'' back to the head to start \textsf{CheckSeg}.}
	\label{fig:DNS1}
\end{figure}

\begin{algorithm}[hbt!]
	\caption{\textsf{DefineSeg}}
	\label{alg:DefNxSg}
	$ S  = ( p_1, \ldots, p_{|S|}) $ is a Hamiltonian shape\\
	Initial configuration: $S \leftarrow S_I$, a  line $ L \subset S $ of length $k = 1, \ldots, \log |S|$, labelled as in Figure \ref{fig:DNS1} topmost	\\
	\SetAlgoLined
	\DontPrintSemicolon
	\vspace{7pt} 
	$\textcircled{\scriptsize H} \gets p_1.c_2 $ \tcp{head sets a  mark  in transmission state}
	\Repeat{\upshape$ (p_{1}.c_2 = \otimes )$} 
	{ 
		\tcp{each agent acts based on its current label state}
		\textbf{Head $l_h$:}    \\
		\lIf{\upshape$ (p_1.c_2 = \textcircled{\scriptsize H})$}
		{$\cdot \gets p_i.c_2$ \tcp{reset transmission state}}
		\vspace{-10px}
		\lIf{\upshape($ p_{i+1}.c_2 = \otimes $)}
		{$\otimes \gets p_1.c_2$ \tcp{observe end mark; end this sub-phase}}
		
		\textbf{Active:} \\
		\If{\upshape($ p_{i-1}.c_2 = \textcircled{\scriptsize H} $) \tcp{observe predecessor with head mark}}
		{$\textcircled{\scriptsize H} \gets p_i.c_2$ \\
			\lIf{\upshape(\textbf{inline} $p_{i}.c_1 = l$)}
			{$\textcircled{\scriptsize L}  \gets  p_i.c_3$}
			\lIf{\upshape(\textbf{tail} $p_{i}.c_1 = l_t$)}
			{$\textcircled{\scriptsize T}  \gets  p_i.c_3$}
		}
	    \lIf{\upshape$ p_{i-1}.c_2 = \textcircled{\scriptsize L} $}
	    {$\textcircled{\scriptsize L} \gets p_i.c_2$ \tcp{predecessor shows inline mark}}
	    \vspace{-10px}
	    \lIf{\upshape$ p_{i-1}.c_2 = \textcircled{\scriptsize T} $}
	    {$\textcircled{\scriptsize T} \gets p_i.c_2$ \tcp{predecessor shows tail mark}}
	    \vspace{-10px}
	    \lIf{\upshape$\big((p_i.c_2 = \textcircled{\scriptsize H}  \vee \textcircled{\scriptsize L} \vee  \textcircled{\scriptsize T}) \wedge p_{i-1}.c_2 = \cdot \big)$}
	    {\\$p_{i}.c_2 \gets p_{i}.c_3$ \tcp{transmit  marks}  $\cdot \gets p_{i}.c_3$     \tcp{rest marks}}
	    \vspace{-10px}
	    \lIf{\upshape$ p_{i+1}.c_2 = \otimes $}
	    {$\otimes \gets p_i.c_2$ \tcp{successor shows end mark}}
	    \vspace{-10px}
	    \lIf{\upshape$ p_{i}.c_2 = \otimes $}
	    {$\cdot \gets p_i.c_2$ \tcp{rest transmission state}}

	    \textbf{Inactive:} \\
	    \lIf{\upshape$ (p_{i-1}.c_2 = \textcircled{\scriptsize H}) $}
	    {$ s_h \gets p_i.c_1$ \tcp{activate to segment head  $ s_h$}}
	    \vspace{-10px}
	    \lIf{\upshape$ (p_{i-1}.c_2 = \textcircled{\scriptsize L}) $}
	    {$ s \gets p_i.c_1$ \tcp{activate to insegment  $ s$}}
	    \vspace{-10px}
	    \lIf{\upshape$ (p_{i-1}.c_2 = \textcircled{\scriptsize T}) $}
	    {\\$ s_t \gets p_i.c_1$ \tcp{activate to segment tail $ s_t$} $\otimes \leftarrow p_i.c_2$ \tcp{initiate  end mark}}
	    \vspace{-10px}
	}	
	\textsf{CheckSeg}
\end{algorithm}
 
\begin{lemma} \label{lem:DefNxSg}
	\textit{DefineSeg} requires at most $O(n)$ rounds to define $S_i$. 
\end{lemma}
\begin{proof}
	The head mark ``$\textcircled{\scriptsize H} $'' shall traverse all agents of the line $L_i$ of length $|L_i|$ until it arrives at the first inactive agent, taking $O(|L_i|)$ rounds. Thus, all other agents on the line propagate marks that take $O(|L_i|)$ parallel rounds to reach at their corresponding agents on the next segment. In the worst case, the line can be of length $n/2$, which requires at most $O(n)$ rounds of communication to identity the next segment $S_i$ of length  $n/2$. \qed
\end{proof}

\subsection{Check the next segment $S_i$}
\label{sec:Check}

This sub-phase checks the geometrical configuration of the new defined segment $S_i$, determining if it is in line with $L_i$, perpendicular to $L_i$ or contains one turn (L-shape). It aims is to save energy in the system, surpassing one or more of the subsequent sub-phases. First, when $S_i$ is in line with $L_i$ (see illustrated in Figure \ref{fig:DNS1}), both $S_i$ and $L_i$ already form a single straight line $L_{i+1}$ of double length, and so the next phase $i+1$ begins. This reduces the cost of \textsf{DrawMap}, \textsf{Push}, \textsf{RecursiveCall} and \textsf{Merge}. Second, $S_i$ is producing a line perpendicular to $L_i$ (see Figure \ref{fig:ConfigOfSeg}(a)), in which case $L_i$ just needs to reverse direction and line up with $S_i$ to generate $L_{i+1}$, proceeding directly to \textsf{PushLine} and avoiding the extra cost of\textsf{DrawMap}. Lastly, $S_i$ has a single turn (looks like L-shape)  (see Figure \ref{fig:ConfigOfSeg}(b)), where it can simply turn at its corner and align with $L_i$, create $L_{i+1}$ and save the cost of \textsf{DrawMap}, \textsf{Push} and \textsf{RecursiveCall}.  A high-level explanation is provided below.

When $l_h$ observes ``$\otimes$'', it propagates its own local direction stored in component $c_4 =  a \in A$ by updating $c_2 \gets c_4$. Then, all active agents on the path forward $a$ from $p_i$ to $p_{i+1}$ via their transmission components. Whenever a $p_i$ with a local direction $c_4 = a^{\prime}  \in A$ notices $a^{\prime} \ne a$, it combines $a$ with its local direction $a^{\prime}$ and changes its transmission component to $c_2 \gets aa^{\prime}$. After that, if a $p_i^{\prime}$ having $c_4 = a^{\prime\prime}  \in A$ observes $a^{\prime\prime} \ne a^{\prime}$, it updates its transmission component into a negative mark, $c_2 \gets \neg$. All signals are to be reflected by the segment tail $s_t$ back to $l_h$, which acts accordingly as follows: (1) starts the next sub-phase \textsf{DrawMap} if it observes ``$\neg$'', (2) calls \textsf{Merge} to combine the two perpendicular lines if it observes $aa^{\prime}$ or (3) begins a new phase $i+1$ if it receives back its local direction $a$. Algorithm \ref{alg:CheckSg}  shows the pseud-code of this sub-routine. 

\begin{algorithm}[hbt!]
	\caption{\textsf{CheckSeg}}
	\label{alg:CheckSg}
	$ S  = ( p_1, \ldots, p_{|S|}) $ is a Hamiltonian shape\\
	Initial configuration: $S \leftarrow S_I$, a  line $ L \subset S $ of length $k = 1, \ldots, \log |S|$, labelled correctly as in Figure \ref{fig:DNS1} bottommost	\\
	\SetAlgoLined
	\DontPrintSemicolon
	\vspace{7pt} 
	$ p_1.c_2 \gets p_1.c_4 $ \tcp{head emits its direction}
	\Repeat{\upshape$p_1.c_2 = \neg$} 
	{ 
		\tcp{each agent acts based on its current label state}
		\textbf{Head $l_h$:}    \\
		\lIf{\upshape$ (p_1.c_2 = c_4)$}
		{$\cdot \gets p_i.c_2$ \tcp{reset transmission state}}
		\vspace{-10px}
		\lIf{\upshape($ p_{i+1}.c_2 = \neg$)}
		{$\neg \gets p_1.c_2$ \tcp{end this sub-phase}}
		\vspace{-10px}
		\lIf{\upshape($ p_{i+1}.c_2 =$ \checkmark)}
		{start phase $i+1$ \tcp{a new phase begins}}
		\vspace{-10px}
		\lIf{\upshape($ p_{i+1}.c_2 =$ \RoundedLsteel)}
		{\textsf{PushLine($L$)} \tcp{$S_i$ has one turn}}
		
		\textbf{Active:}    \\
		\lIf{\upshape($ p_{i-1}.c_2 = c_4)$}
		{$c_4 \gets p_i.c_2$ \tcp{observe same direction}}
		\vspace{-10px}
		\lIf{\upshape$ p_{i-1}.c_2 \ne c_4$)}
		{$c_4$\RoundedLsteel \hspace{1px}$\gets p_i.c_2$ \tcp{show a turn}}
		\vspace{-10px}
		\lIf{\upshape$(p_{i-1}.c_2 = c_4$\RoundedLsteel)}
		{$\neg \gets p_i.c_2$ \tcp{show another turn}}
		\vspace{-10px}
		\lIf{\upshape$ (p_{i+1}.c_2 = \neg$ $ \vee$ \checkmark \hspace{2px}$\vee$ \RoundedLsteel )}
		{$p_i.c_2 \gets p_{i+1}.c_2$ \tcp{transmit marks backwards}}
		\vspace{-10px}
		\lIf{\upshape$ (p_{|2L|-1}.c_2 =  c_4)$}
		{\checkmark $\gets p_{|2L|.c2}$\tcp{$s_i$ transmits mark backwards}}
		\vspace{-10px}
		\lIf{\upshape$ (p_{|2L|-1}.c_2 =$  $c_4$\RoundedLsteel)}
		{\RoundedLsteel $\gets p_{|2L|.c2}$ \tcp{$s_i$ transmits mark backwards}}
		\vspace{-10px}
		\lIf{\upshape$ (p_{i-1}.c_2 \ne \cdot)$}
		{$\cdot \gets p_i.c_2$  \tcp{reset transmission state}}
		\vspace{-10px}	
	}	
	\textsf{DrawMap} 
\end{algorithm}

\begin{figure}[hbt!]
	\centering 
	\subfigure[$L_i$ is perpendicularly to $S_i$.]{\includegraphics[scale=1.2]{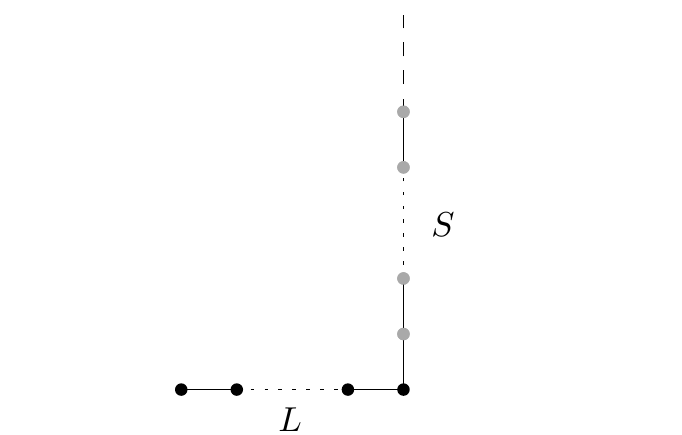}}\qquad \qquad 
	\subfigure[$S_i$ has a single turn.]{\includegraphics[scale=1.2]{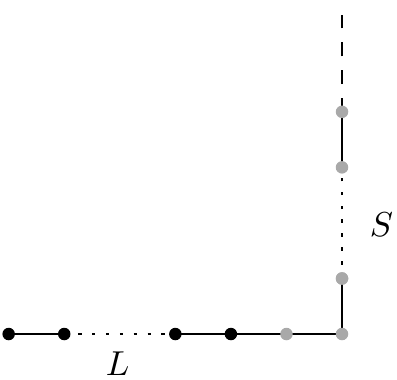}}
	\caption[Two configurations of a Hamiltonian path.]{Two configurations of a Hamiltonian path terminates at a straight line $L_i$ (in black dots) followed by a segment $S_i$ (in grey dots) on the path.}
	\label{fig:ConfigOfSeg}
\end{figure}

\begin{lemma} \label{lem:CheckNextSegCorrection}
	\textsf{CheckSeg} correctly checks the configuration of $S_i$  to be one of the following:
	\begin{itemize}
		\item $S_i$ is in line with $L_i$.
		\item $S_i$ forms a straight line perpendicular to $L_i$.
		\item $S_i$ forms an L-shape.
		\item $S_i$ contains more than one turn. 
	\end{itemize} 
\end{lemma}
\begin{proof}
	This sub-routine starts as soon as the head  $l_h$ observes the end mark ``$\otimes$'' of \textsf{DefineSeg}, which means that all agents of the segment $S_i$ are active and labelled correctly. Given that, the input configuration of \textsf{CheckSeg} is a Hamiltonian path terminates at straight line $L_i$ followed by $S_i$, both are composed of $2^i$ active agents. All other inactive agents in the rest of the configuration are labelled $k$.  During this sub-phase, the active agents use their local directions  of the path stored in state $c_4$ in which a $p_i$ knows each ports incident to predecessor  $p_{i-1}$ and successor $p_{i+1}$.
	
    Now, $l_h$ updates it transmission state to $c_2 \gets c_4$ where it emits its local direction held in $c_4$.  Assume without loss of generality, $c_4$ holds a local direction pointing to the east neighbour ``$ \rightarrow $'', then $l_h$ performs this state transition: $ \delta(l_h, \otimes,\cdot, \rightarrow) = (l_h, \rightarrow,\cdot, \rightarrow)$. This arrow ``$\rightarrow$" propagates through transmission states to all active agents of $L_i$ and $S_i$. When a $p_{i-1}$ displays an empty transmission state, each agent $p_i$ updates state $c_2$ to ``$\cdot$''.  If ``$\rightarrow$'' matches a local direction stored on $c_4$ of  $p_i$, then $p_i$ transmits the same arrow ``$\rightarrow$'' from $p_{i-1}$ to  $p_{i+1}$. If $p_i$ stores a turning arrow (e.g. ``$\downarrow$'' or ``$\uparrow$'')  on $c_4$, then it updates state $c_2$ with a special L-shape mark, ``$\rightarrow$\RoundedLsteel'', which is then passed to $p_{i+1}$. Whenever  $p_j$ stores a turning arrow and observes $p_{j-1}$ showing ``$\rightarrow$\RoundedLsteel'' , $p_j$ initiates a negative mark $c_2 \gets \neg $, which is relayed back to $l_h$, calling out for \textsf{DrawMap}. Once $s_t$ observes ``$\rightarrow$\RoundedLsteel'', it bounces off the mark ``\RoundedLsteel'' back towards $l_h$ to start \textsf{PushLine}.  Otherwise, $s_t$ propagates a special check-mark ``\checkmark''  backwards, alerting $l_h$ that both $L_i$ and $S_i$ already form a straight line.  \qed
\end{proof}

Now, we provide analysis of this procedure.

\begin{lemma} \label{lem:CheckSeg}
	An execution of \textsf{CheckSeg} requires at most $O(n)$ rounds of communication. 
\end{lemma}
\begin{proof}
	Consider the worst-case in which the direction mark traverses a $n$-length path and a special mark ``\checkmark'' bounces off the other end of the path and returns to at the head. This journey takes at most $2n-2$ rounds, during which an agent $p_i$, $1\le i \le n$, emits the direction mark to $p_{i+1}$ and ``\checkmark'' to $p_{i-1}$, with the exception of the two endpoints of the path. \qed
\end{proof}

\subsection{Draw a route map}
\label{sec:Draw}

This local procedure generates a map and computes the shortest possible route with the fewest turns, aiming for the lowest cost and energy consumption. On the square grid, the most efficient way to accomplish this is to draw a rout of a single turn, such as L-shaped routs. For the purpose of connectivity preservation, it can be demonstrated that there exist some worst-case routes in which the line $L_i$ may disconnect while travelling towards the tail of $S_i$. Essentially, this can be seen in a route where the Manhattan distance between the the line tail $l_t$ and the segment tail $s_t$ is $\Delta(l_t, s_t) \ge |L_i|$, for additional information, see \cite{AMP2020}.

Thus, this distance $\Delta(l_t, s_t)$ is important in determining whether to take an L-shaped route directly to $s_t$ or to go through an intermediary agent of $S_i$ who passes through two L-shaped routes. From our distributed perspective, the Manhattan distance $\Delta(l_t,s_t)$ cannot be computed in a straightforward manner due to several challenges, such as individuals with constant local memory and limited computational power. Below, \textsf{DrawMap} addresses these challenges by using $L_i$  agents as (1) a distributed binary counter for calculating the distance and (2) a distributed memory for storing local directions of agents, which collectively draw the route map.

This sub-phase computes the Manhattan distance $\Delta(l_t,s_t)$ between the line tail $l_t$ and the segment tail $s_t$, by exploiting \textsf{ComputeDistance} in which the line agents implement a distributed binary counter. First, the head $l_h$ broadcasts ``$\textcircled{\scriptsize C}$'' to all active agents, asking them to commence the calculation of the distance. Once a segment agent $p_i$ observes ``$\textcircled{\scriptsize C}$'', it emits one increment mark ``$\oplus$'' if its local direction is cardinal or two sequential increment marks if it is diagonal. The ``$\oplus$'' mark is forwarded from $p_i$ to $p_{i-1}$ back to the head $l_h$. Correspondingly, the line agents are arranged to collectively act as a distributed binary counter, which increases by 1 bit per increment mark, starting from the least significant at $l_t$.  

When a line agent observes the last ``$\oplus$'' mark, it sends a special mark ``$\textcircled{\scriptsize 1}$'' if $\Delta(l_t,s_t) \le |L_i| $ or ``$\textcircled{\scriptsize 2}$'' if $\Delta(l_t,s_t) > |L_i| $ back to $l_h$. As soon as $l_h$ receives ``$\textcircled{\scriptsize 1}$'' or ``$\textcircled{\scriptsize 2}$'', it calls \textsf{CollectArrows} to draw a route that can be either heading directly to $s_t$ or passing through the middle of $S_i$ towards $s_t$. In \textsf{CollectArrows}, $l_h$ emits ``\counterplay'' to announce the collection of local directions (arrows) from $S_i$. When ``\counterplay'' arrives at a  segment agent, it then propagates its local direction stored in $c_4$ back towards $l_h$. Then, the line agents distribute and rearrange $S_i$'s local directions via several primitives, such as cancelling out pairs of opposite directions, priority collection and pipelined transmission. Finally, the remaining arrows cooperatively draw a route map for $L_i$ (see Definition \ref{def:route}). The following lemma shows that this procedure calculates $\Delta(l_t, s_t)$ in linear time. Below, we give more details of \textsf{DrawMap}. 

\begin{definition}[A route] \label{def:route} 
	A route is a rectangular path $R$ consisting of a set of cells $R = [c_1, \ldots, c_{|R|}]$ on $\mathbb{Z}^2$, where $c_i$ and $c_{i+1}$ are two cells adjacent vertically or horizontally, for all $1 \le i \le |R|-1$. Let $C$ be a system configuration, $C_R$ denotes the configuration of $R$ where $C_R \subset C$ defined by $[c_1, \ldots, c_{|R|}]$. 
\end{definition}

\subsubsection{Distributed Binary Counter.}

Due to the limitations of this model, individual agents cannot calculate and keep non-constant numbers in their state. Alternatively, the line $L_i$ of $k = 2^i$ agents can be utilised as a distributed binary counter (similar to a Turing machine tape) which is capable to store up to $2^k-1$ unsigned values. This $k$-bit binary counter supports increment which is the only operation we need in this procedure. Each agent's counter state $c_5$ is initially ``$\cdot$''  and can then hold a bit from $\{0,1\}$. The line tail $l_t$ denotes the least significant bit of the counter.  An increment operation is performed as follows: Whenever a line agent $p_i$ detects $p_{i+1}$ showing an increment mark ``$\oplus$'', $p_i$ switches counter component $c_5$ from ``$\cdot$'' or 0 to 1 and destroys the ``$\oplus$'' mark. If $p_i$ holds 1 in $c_5$, it flips 1 to  0 and redirects the increment mark ``$\oplus$'' to $p_{i-1}$ (i.e. update the transmission state $c_2$ to ``$\oplus$''). See an implantation of this counter in Figure \ref{fig:Counter}.

\begin{figure}[hbt!]
	\centering 
	\includegraphics[scale=0.72]{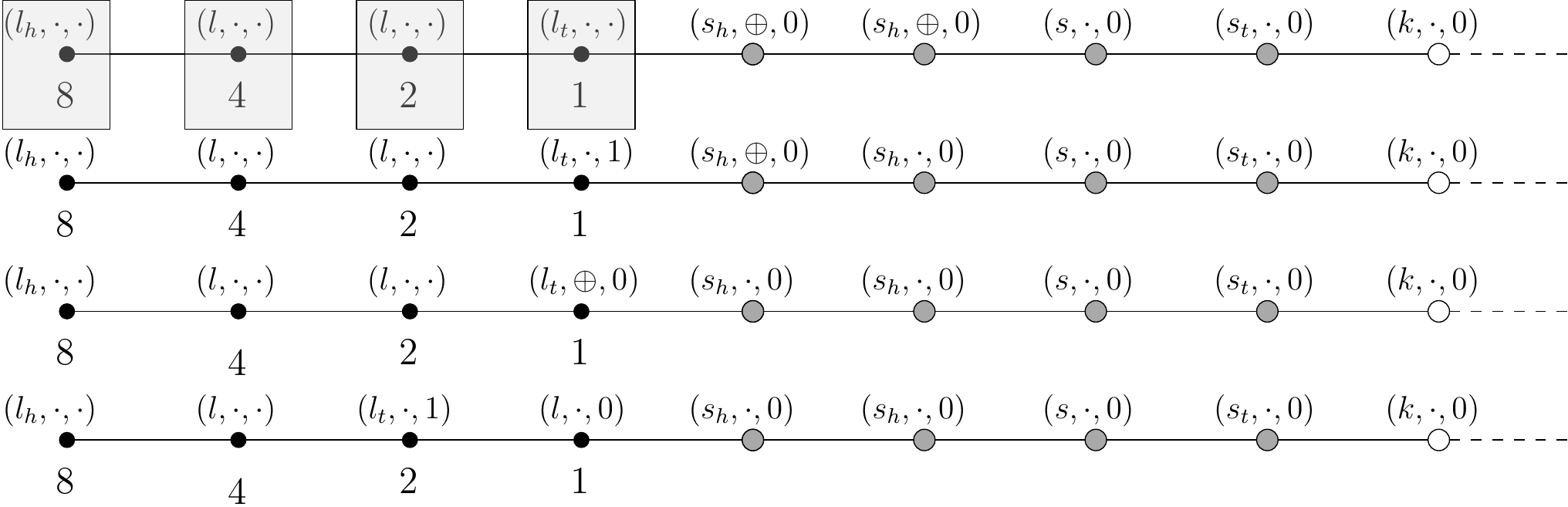}
	\caption[An implantation of 4-bit line counter $L_i$.]{A 4-bit line counter $L_i$. Agents of $L_i$ and $S_i$ are depicted by black grey dots, respectively. The state of an agent is $(c_1, c_2, c_5)$ denoting $c_1$ the label, $c_2$ transmission and $c_5$ counter components, omitting others with no effect. Each shaded area shows a corresponding decimal number.  Top: the counter has a decimal value of 0. 2nd: an increment of 1. 3rd: the line tail $l_t$ flips state $c_5$ from 1 to 0 and updates $c_2$ with ``$\oplus$''. Bottom: the counter increased by 1 corresponding to a decimal value 2.}
	\label{fig:Counter}
\end{figure}

\subsubsection{\textsf{ComputeDistance} procedure}

Initially, the head $l_h$ emits a special mark ``$\textcircled{\scriptsize C}$'' to all active agents, asking them to commence the calculation of the Manhattan distance $\Delta(l_t, s_t)$ between the line tail $l_t$ and the segment tail $s_t$. Whenever a segment agent $p_i$ (of label $s_h, s$ or $ s_t$) observes $p_{i-1}$ with ``$\textcircled{\scriptsize C}$'', it performs one of two transitions: (1) It updates transmission state to $c_2 \gets \oplus$ if its local direction stored in $c_4$ is cardinal (horizontal or vertical) from $\{ \rightarrow, \leftarrow, \uparrow, \downarrow \}$, (2) if $c_4$ holds a diagonal direction from $\{ \nwarrow, \nearrow, \swarrow, \searrow \}$, it receptively updates the transmission and waiting states, $c_2$ and $c_3$, to ``$\oplus$''. Eventually, the segment head $s_h$ produces the last special increment mark ``$\oplus^{\prime}$''. In principle, any diagonal direction between two cells in a square grid can increase the distance by two (in the Manhattan distance), whereas horizontal and vertical directions always increase it by one.
 
As a result, all increment marks initiated by segment agents are transmitted backwards to the counter $L_i$, in the same way that the propagation of an end mark is described in \textsf{DefineSeg}. Hence, the binary counter increases by 1 bit each time it detects ``$\oplus$'', starting from the least significant bit stored in $l_t$. Because of transmission parallelism, the binary counter may increase by more than one bit in a single round. When a line agent $p_i$ sees predecessor with the last increment mark ``$\oplus^{\prime}$'', $p_i$ passes ``$\textcircled{\scriptsize 1}$'' towards the line head $l_h$. This mark ``$\textcircled{\scriptsize 1}$''  is altered to ``$\textcircled{\scriptsize 2}$'' on its way to $l_h$ only if it passes a line agent of a counter state $c_5 = 1$, otherwise it is left unchanged. Eventually, the head $l_h$ observes either ``$\textcircled{\scriptsize 1}$'', by which it calls \textsf{CollectArrows} procedure to draw a route map directly to the tail $s_t$ of $S_i$, or ``$\textcircled{\scriptsize2}$'', by which it calls \textsf{CollectArrows} to push via a middle agent $s$ towards $s_t$. We provide Algorithm \ref{alg:ComputeDistance} of the \textsf{ComputeDistance} procedure below.

\begin{algorithm} [hbt!] 
	\caption{\textsf{ComputeDistance($L_i, S_i$)}} 
	\label{alg:ComputeDistance} 
	$ S  = ( p_1, \ldots, p_{|S|}) $ is a Hamiltonian shape\\
	Initial configuration: a straight line $L_i$ and a segment $S_i$ labelled as in Figure \ref{fig:Counter} \\
	\SetAlgoLined
	\DontPrintSemicolon
	\vspace{7pt} 
	1. The line head $l_h$ propagates counting mark $\textcircled{\scriptsize C}$ along $L_i$ and $S_i$

	2. Once $\textcircled{\scriptsize C}$ arrives at the segment tail $s_t$, a segment agent acts as follows: 
	
	3. $s_t$ sends one increment $\oplus$ back to $l_h$ if its direction is cardinal or two $\oplus$ if diagonal 
	
	\tcp{pipelined transmission}
	
	4a. $s$ observes $\oplus$, sends one increment $\oplus$ back to $l_h$ if its direction is cardinal or two $\oplus$ if diagonal 
	
	4b. $s_l$ observes $\oplus$, sends one increment $\oplus^{\prime}$ back to $l_h$ if its direction is cardinal or two $\oplus^{\prime}$ if diagonal 
	
	5. The distributed counter $L_i$ increases by 1 bit each time it receives  $\oplus$
	
    6. A line agent observes the last $\oplus^{\prime}$ coming to $L_i$, it sends  a mark $\textcircled{\scriptsize 1}$ back to $l_h$ 
    
    7a. Each line agent observes $\textcircled{\scriptsize 1}$ and has 1 bit, it passes $\textcircled{\scriptsize 2}$ towards $l_h$ 
    
    7b. Each line agent observes $\textcircled{\scriptsize 1}$ and has 0 bit, it passes $\textcircled{\scriptsize 1}$ towards $l_h$ 

	7c. Each line agent observes $\textcircled{\scriptsize 2}$, it passes $\textcircled{\scriptsize 2}$ towards $l_h$
	
	\tcp{Manhattan distance $\Delta \le i $}
	8a. When $l_h$ sees $\textcircled{\scriptsize 1}$, it calls  \textsf{CollectArrows} to draw one L-shaped route 
	
	\tcp{Manhattan distance $\Delta >  i $}		
	8b. Otherwise, $l_h$ sees $\textcircled{\scriptsize 2}$ and calls \textsf{CollectArrows} to draw two L-shaped route 
\end{algorithm}

Let $\Delta(l_t, s_t)$ denote the Manhattan distance between the line tail $l_t$ and the segment tail $s_t$. The following lemma shows that this procedure calculates $\Delta(l_t, s_t)$ in linear time. 

\begin{lemma} \label{lem:CompDis}
	\textsf{ComputeDistance} requires $O(|L_i|)$ rounds to compute $\Delta(l_t, s_t)$.
\end{lemma}
\begin{proof}
	Consider an input configuration labelled $(\overbrace{l_h, \ldots, l, \ldots, l_t}^{L_i}, \overbrace{s_h,\ldots, s, \ldots, s_t}^{S_i},$ $ k, \ldots, k)$, starting at a line head $p_1$ of label $l_h$, where $|L_i| = |S_i|$. We only show affected states in this proof. Initially, $l_h$ emits a counting mark ``$\textcircled{\scriptsize C}$'' by updating transmission state to $p_1.c_2 \gets \textcircled{\scriptsize C}$, then $l_h$ resets  transmission state to $ c_2 \gets \cdot $ in subsequent rounds.  Once an active agent $p_i$ in round $r_{j-1}$ (where $j \le 2|L_i|$) detects predecessor showing state $p_{i-1}.c_2 = \textcircled{\scriptsize C}$, it updates transmission state to $p_i.c_2 \leftarrow \textcircled{\scriptsize C}$ in $r_j$ and then resets $p_i.c_2 \gets \cdot$ in $r_{j+1}$.  Upon arrival of  ``$\textcircled{\scriptsize C}$'' at $s_t$, its predecessor changes transmission state to $c_2 \leftarrow\oplus$ and puts another increment mark in waiting state $c_3 \leftarrow\oplus$ if it stores a diagonal arrow in its local direction  $c_4$.
	
	Due to the goal of counting, the direction of $s_t$ is dropped. Each segment agent $p_i$ of label $s_h$ and $s$ observes a successor presenting state $p_{i+1}.c_2 = \oplus$ in round $r_{j-1}$, then the following transitions apply in $r_{j}$: (1) $p_i.c_2 \leftarrow \oplus$ if $p_{i+1}.c_2 \leftarrow \oplus$, (2) if $p_i.c_2 \leftarrow \oplus$ if $p_{i+1}.c_2 \leftarrow \cdot$ and $p_i.c_3\leftarrow \oplus$, (3) the head of segment $s_h$ sets  $p_i.c_2 \leftarrow \oplus^{\prime}$ if $p_{i+1}.c_2 \leftarrow \cdot$ and $p_i.c_3\leftarrow \oplus$ and (4)  $p_i.c_2 \leftarrow \cdot$ if $p_{i+1}.c_2 \leftarrow \cdot$ and $p_i.c_3\leftarrow \cdot$.
	
	Correspondingly, the line agents (of labels $l_h$, $l$ and $l_t$) behave as a binary counter described above and illustrated in Figure \ref{fig:Counter}. When a line agent $p_i$ detects ``$\oplus$'' in the state of $p_{i+1}$ in round $r_{j-1}$, it updates state based on one of theses transitions in round $r_{j-1}$: (1) $p_i.c_5 \leftarrow 1$ if $p_i.c_5 \leftarrow \cdot$ or $p_i.c_5 \leftarrow 0$ or (2) $p_i.c_5 \leftarrow 0$ and $p_i.c_2 \leftarrow \oplus$  if  $p_i.c_5 \leftarrow 1$. In the case where the last increment mark ``$\oplus^{\prime}$'' detected by $p_i$ in round $r_{j-1}$, then $p_i$ updates state to $p_i.c_2 \leftarrow \textcircled{\scriptsize 1}$ in $r_j$. When $p_{i-1}$ observes $\textcircled{\scriptsize 1}$, then it updates states to either (1) $p_{i-1}.c_2 \leftarrow \textcircled{\scriptsize 1}$ if $p_{i-1}.c_5 = 0$ or (2) $p_{i-1}.c_2 \leftarrow \textcircled{\scriptsize 2}$ if $p_{i-1}.c_5 = 1$. Thus, the mark ``$\textcircled{\scriptsize 2}$'' is sent back to the head $l_h$, which finally sees either `$\textcircled{\scriptsize 1}$'' or ``$\textcircled{\scriptsize 2}$'' and acts appropriately (calls \textsf{CollectArrows} procedure). The counter size is sufficient to calculate $\Delta(l_t, s_t)$ because the since the worst-case would distance is $|L_i| - 2$. 
	
	Now, we analyse the cost of communication of this procedure in a number of rounds. First, the counter mark ``$\textcircled{\scriptsize C}$'' goes on a journey that takes $t_1 = 2|L_i|= O(|L_i|)$ rounds. That is, the pipelined transmission of increment marks requires at most $t_2 = O(|L_i|)$ parallel rounds of communication. Moreover, the marks ``$\textcircled{\scriptsize 1}$'' or ``$\textcircled{\scriptsize 2}$'' travel to the head $l_h$ within at most $t_3 = O(|L_i|)$. Altogether, the total running time is bounded by $t  = t_1 +  t_2 + t_3 = O(|L_i|)$ parallel rounds. \qed
\end{proof}

\subsubsection{\textsf{CollectArrows} procedure}

Informally, the distance obtained from the \textsf{ComputeDistance} procedure can be (1) equal or less than the line length $|L_i|$  ($l_h$ observes this mark ``$\textcircled{\scriptsize 1}$'') or (2) greater than $|L_i|$ ($l_h$ observes ``$\textcircled{\scriptsize 2}$'').  In case (1), it propagates a special collection mark ``\counterplay'' through all active agents until it reaches the segment tail $s_t$. When ``\counterplay'' arrives, $s_t$ broadcasts its local arrow in $c_4$ back to $l_h$ via active agent transmission states. This journey accomplishes the following: (a) Gathers arrows similar to $s_t$ and puts them in priority transmission. (b) Eliminates pairs of opposite arrows and replaces them with a hash mark ``\#''. (c) Arranges the arrows on $L_i$'s distributed memory. In case (2), $l_h$ emits a special mark  ``$ \textcircled{\scriptsize M} $''  to  $s_h$, defining a midpoint on $S_i$ through which the line $L_i$ passes towards $s_t$.

Now, $s_h$ propagates two marks down $s_t$, a fast mark ``$ \textcircled{\tiny m1} $'' is transmitted every round and a slow mark moves three rounds slower ``$ \textcircled{\tiny m2} $''. The fast mark ``$ \textcircled{\tiny m1} $''  bounces off $s_t$, where both ``$ \textcircled{\tiny m1} $'' and ``$ \textcircled{\tiny m2} $'' meet in a $S_i$ middle agent $p_j$, which changes label to $s_t^{\prime}$ and a successor $p_{j+1}$ switches to $s_h^{\prime}$. This temporally divides $S_i$ into two segments, $S_i^1= s_h, \ldots,  s_t^{\prime}$ and $S_i^2 = s_h^{\prime}, \ldots , s_t $. The middle agent $s_t^{\prime}$ propagates ``$ \textcircled{\scriptsize M} $'' to tell $l_h$ that a midpoint has been identified. Case (1) is then repeated twice to collect arrows from  $S_i^1$ and $S_i^2$ and distribute them into the line agents (distributed memory). After that,  \textsf{Push($S$)} begins. Algorithm \ref{alg:CollectArrows} presents the pseudocode that briefly formulates this procedure.

\begin{algorithm}[hbt!] 
	\caption{\textsf{CollectArrows($L_i, S_i$)}} 
	\label{alg:CollectArrows} 
		 Input: a straight line $L_i$ and a segment $S_i$
		 \SetAlgoLined
		 \DontPrintSemicolon
		 \vspace{7pt} 
		 
		 \tcp{{\small priority and pipelined transmission, see text for details}}
		\underline{(A) Line head $l_h$ observes $\textcircled{\scriptsize 1}$}
		
		1.  $l_h$ propagates collection mark \counterplay 
		
		 2. Each active agent $p_{i}$ emits \counterplay \hspace{1px} to $p_{i+1}$ 
		  
		3. $s_t$ observes $\textcircled{\scriptsize 1}$ and propagates its direction $d$ in $c_4$, $c_2 \leftarrow c_4$ 
		
		 4. Each segment agent $p_{i}$ passes a direction to $p_{i-1}$ 
		
		5. Distribute directions into the line agents
	    
	    6. Rearrangement of directions
		
		7. \textsf{Push($S$)} begins \\
		
		 \vspace{7pt} 
		\underline{(B) Line head $l_h$ observes $\textcircled{\scriptsize 2}$}
		 
		 1.  $l_h$ propagates a midpoint mark $ \textcircled{\scriptsize M} $
		
		2. Each line agent $p_{i}$ broadcasts $ \textcircled{\scriptsize M} $ to $p_{i+1}$ 
		
		3. $s_h$ sees $ \textcircled{\scriptsize M} $, then emits fast  $ \textcircled{\tiny m1} $ and slow $ \textcircled{\tiny m2} $ waves down to  $s_t$
		
		4. $ \textcircled{\tiny m1} $ bounces off $s_t$  and meets $ \textcircled{\tiny m2} $ at middle agent $p_j$ with label changed to $s_t^{\prime}$
		
		5. $s_t^{\prime}$ propagates $ \textcircled{\scriptsize M} $ to $l_h$
		
		6. Once $l_h$ sees $ \textcircled{\scriptsize M} $ again, it goes to (A)   
\end{algorithm}

The following lemma proves the correctness and analysis of \textsf{CollectArrows}.

\begin{figure}[hbt!]
	\centering
	\includegraphics[scale=0.8]{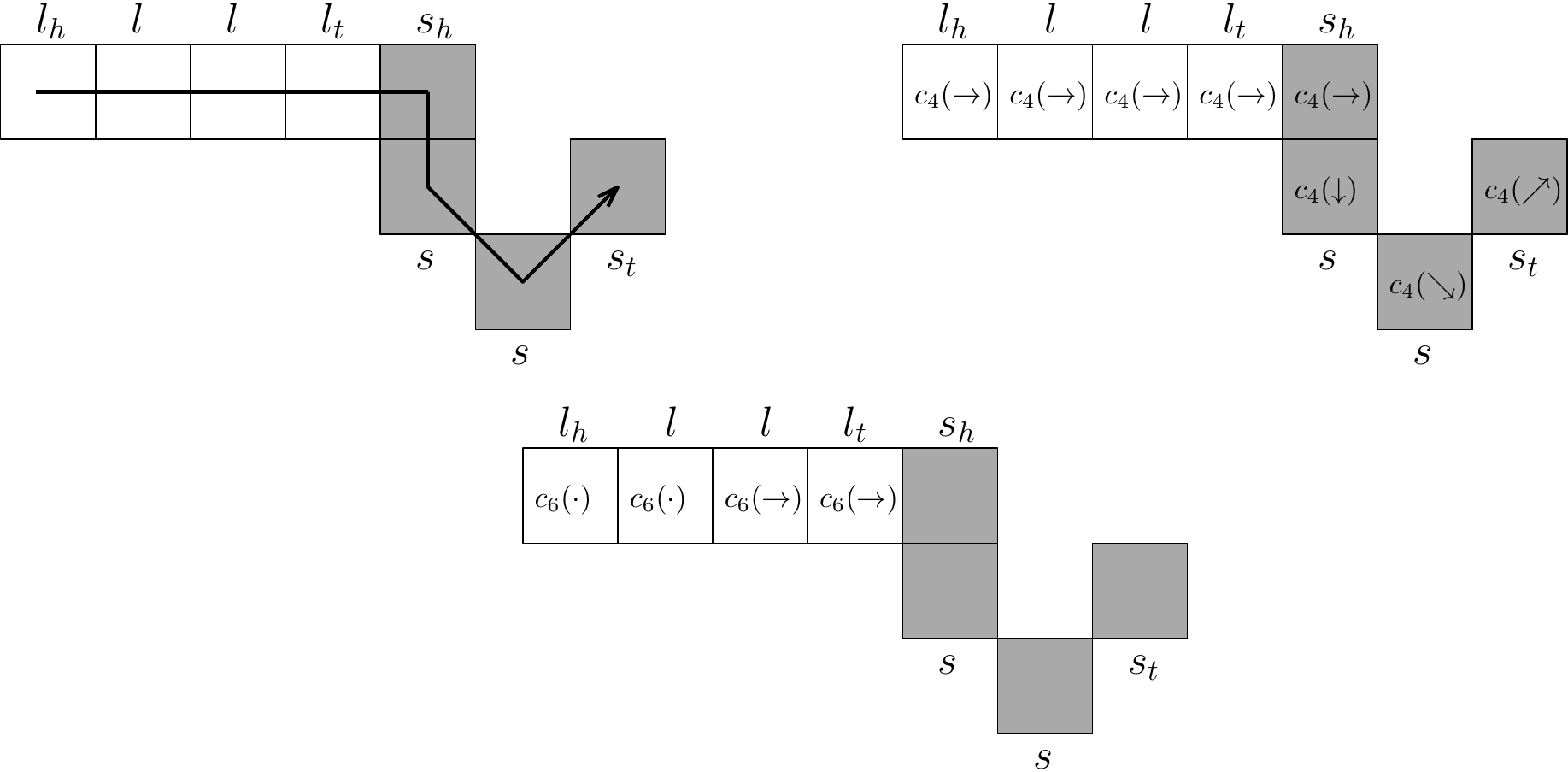}
	\caption[Drawing a map: A path and corresponding local arrows.]{Drawing a map: from top-left a path across occupied cells  and corresponding local arrows stored on state $c_4$ in top-tight, where the diagonal directions, ``$\searrow$'' and ``$\nearrow$'', are interpreted locally as, ``$\downarrow\rightarrow$'' and ``$\uparrow\rightarrow$''. The bottom shows a route map drawn locally on state $c_6$ of each line agent.}
	\label{fig:ArrowAndDirections}
\end{figure}

\begin{lemma} \label{lem:CollectArrows}
	The \textsf{CollectArrows} procedure completes within $O(|L_i|)$ rounds.
\end{lemma}
\begin{proof}
	Given an initial configuration defined in Lemma \ref{lem:CompDis}. Assume the Manhattan distance $\delta(l_t, s_t) \le |L_i|$. For simplicity, we prove (A) in algorithm \ref{alg:CollectArrows} showing only affected states. Once $l_h$ observes $\textcircled{\scriptsize 1}$, it emits a collection mark ``\counterplay'', which then transfers forwardly among active agents until it reaches $s_t$, similar to counting mark transmission described previously in Lemma \ref{lem:CompDis}. When  $s_t$ detects ``\counterplay'', it updates transmission state $c_2$ with its local direction held in  $c_4= d$; recall that $d$  is an arrow that locally shows where the Hamiltonian path comes in and out, $d \in \{\rightarrow, \leftarrow, \downarrow, \uparrow, \nwarrow, \nearrow, \swarrow,\searrow\}$.
	
	In what follows, we distinguish between cardinal $\{\rightarrow, \leftarrow, \downarrow,\uparrow\}$ and diagonal directions $\{ \nwarrow, \nearrow, \swarrow, \searrow\}$. Figure \ref{fig:ArrowAndDirections} shows how local arrows are assigned to agents according to the Hamiltonian path. For a cardinal local direction, $s_t$ updates transmission state to $c_2 \gets  d$ and marks local direction state with a star $c_4 \gets d^{\star}$, indicating that $d$ has been collected. A diagonal local direction between any two neighbouring cells on the two-dimensional square grid is made up of two cardinal arrows, such as $\nwarrow$ is composed of  $\uparrow$ and $\leftarrow$. In other words, an agent needs to move two steps to occupy an adjacent diagonal cell. For example, if  $s_t$ stores a diagonal direction in $c_4$, it puts $d^1$ on transmission $c_2 \gets d^1$, $d^2$ on waiting state $c_3 \gets d^2$,  and marks it with a star, $c_4 \gets d^{\star}$. Next round, the transmission state of $s_t$ resets $c_2 \gets \cdot$ if $c_3$ is empty or sets $c_2 \gets c_3$ if  $c_3$ contains an arrow.
		
	We now show the priority and pipelined collection of local arrows of $S_i$ (in Algorithm \ref{alg:CollectArrows}). Assume a direction (arrow) $d^{+}$ transmits from the segment tail $s_t$, travelling through transmission states via an active agent $p_{i+1}$ to $p_i$. When $d^{+}$ encounters an opposite arrow $d^{-}$ recorded in transmission state $p_i.c_2$, both are erased and replaced by the hash sign ``\#''.  If $d^{+}$ and $d^{+}$ are similar, both take priority in $c_2$. If  $d^{+}$ observes a perpendicular arrow $\perp d$, $d^{+}$ is placed in $c_2$ and $\perp d$ in waiting state $c_3$. For example, Figure \ref{fig:ExCollectArrow} depicts a configuration of $S_i$ consisting of 8 agents, the arrows of which are collected in Figure \ref{fig:CollectArrows}. Full details for the associated transitions per active agent are provided below. 
		
	Given a segment agent $p_i$ of label $s$ and $s_t$ in round $r_{j-1}$, where $j \le 2|L_i|$. Then we show how $p_i$ acts when the direction is either cardinal or diagonal. Consider $p_i$ of an uncollected cardinal direction $d_i$ observes $p_{i+1}$ showing a direction $d_{i+1}$, two directions $d_{i+1} (d_{i+1}^1d_{i+1}^2)$  or \# in transmission component $c_2$. Then, $p_i$ updates its state in $r_j$ as follows: (1) Set $d_{i+1}$ or $ d_{i+1}^1 $  in transmission $p_i.c_2 \leftarrow p_{i+1}.c_2$, put $d_{i}$ in waiting $p_i.c_3 \leftarrow p_i.c_4$ and mark it $p_i.c_4 \leftarrow d_{i}^{\star}$ if $d_{i}$ is perpendicular to $d_{i+1}$, such as $\rightarrow$ and $\uparrow$. (2) Set $p_i.c_2 \leftarrow $ \#, put $d_{i}$ in waiting $p_i.c_3 \leftarrow p_i.c_4$ and mark its local direction $p_i.c_4 \leftarrow d_{i}^{\star}$ if $d_{i}$ and $d_{i+1}$ are a pair of opposite arrows, such as $\uparrow$ and $\downarrow$.  (3) Set both directions $d_{i+1}$ and $d_{i}$ in transmission $p_i.c_2 \leftarrow d_{i+1} d_{i}$, resets $c_3 \leftarrow \cdot$ and  mark $d_{i}$ with a star $p_i.c_4 \leftarrow d_{i}^{\star}$ if $d_{i}$ and $d_{i+1}$ are a pair of same arrows, such as $\uparrow$ and $\uparrow$. When a cardinal direction is already collected $d_i^{\star}$,  $p_i$ sets $d_{i+1}$ or $ d_{i+1}^1 $  in transmission $p_i.c_2 \leftarrow p_{i+1}.c_2$. If $d_{i+1}$ ( or $d_{i+1}^1 $)  and $ c_3 = d_i$ are similar, then $p_i$ sets $p_i.c_2 \leftarrow d_{i+1}d_i $ (or $p_i.c_2 \leftarrow d_{i+1}^1 d_i $) and resets $p_i.c_3 \leftarrow \cdot$.	 If $p_{i+1}.c_2$ is empty, then $p_i$  puts waiting direction in transmission $p_i.c_2 \leftarrow p_i.c_4$ or rests $p_i.c_2 \leftarrow \cdot$, otherwise. 
	
	In the second case, $p_i$ holds an uncollected diagonal arrow $d_i(d_i^1d_i^2)$ in $r_{j-1}$, so it performs one of the following in $r_j$: (1) Set $d_{i+1}$ and $d_i^1$ in transmission $p_i.c_2 \leftarrow d_{i+1}d_i^1$, put $d_i^2$ in waiting $p_i.c_3 \leftarrow d_i^2$ and mark $d_{i}$ with a star $p_i.c_4 \leftarrow d_{i}^{\star}$ if $d_{i+1}$ and $ d_i^1$ (or $ d_i^2 )$ are similar, such as $\uparrow$ and $\nwarrow= (\uparrow\leftarrow)$. (2) Set $p_i.c_2 \leftarrow $ \#, put $d_i^2$ in waiting $p_i.c_3 \leftarrow d_i^2$ and mark the direction  $d_i^{\star} $ if $d_{i+1}$ is opposites to either  $ d_i^1$ or $ d_i^2 $, such as $\uparrow$ and $\swarrow =(\downarrow\leftarrow)$. If a diagonal arrow has been already collected $d_i^{\star}$, then $p_i$ sets $d_{i+1}$ or $ d_{i+1}^1 $  in transmission $p_i.c_2 \leftarrow p_{i+1}.c_2$. If $d_{i+1}$ (or $d_{i+1}^1 $)  and waiting direction $ c_3 = d_i$ are the same, then $p_i$ updates to $p_i.c_2 \leftarrow d_{i+1}d_i $ (or $p_i.c_2 \leftarrow d_{i+1}^1 d_i $)  and resets $p_i.c_3 \leftarrow \cdot$.	 If $p_{i+1}.c_2$ is empty then, $p_i$  puts waiting direction in transmission $p_i.c_2 \leftarrow p_i.c_4$ or rests $p_i.c_2 \leftarrow \cdot$, otherwise. 
	
	Meanwhile, the line agents receive the collected arrows and divide them among respective states as follows. Let $p_i$ denote a line agent, holding a map state $p_{i}.c_6 =\cdot$,  observes $p_{i+1}$ showing a direction $ d_{i+1}$ or a hash sign ``\#''. Then, $p_i$ acts accordingly: (1)  $p_{i}.c_6 \gets d_{i+1}$, (2) if $p_i$ is $l_h$ or sees $p_{i-1}$ with a map state $c_6= $\#, then  $p_{i}.c_6 \gets $\#. Whenever $p_{i}.c_6 \ne\cdot$ detects $p_{i+1}.c_2 = d$ or $p_{i+1}.c_2 =$\# , then $p_i$ updates state to  $d_{i+1}$ or ``\#"  if $p_{i-1}.c_6= \cdot$. 	Once the line tail $l_t$  of a non-empty map component detects $p_{i+1}.c_2  = \cdot$, it propagates a special mark ``\counterplay\checkmark'' via line agents towards $l_h$, announcing the completion of arrows collection. 
	
	Now, let us discuss (B) in algorithm \ref{alg:CollectArrows} in which $l_h$ observes ``$\textcircled{\scriptsize 2}$'', indicating the Manhattan distance $\delta(l_t, s_t) > |L_i|$. In reaction to this, $l_h$ emits the midpoint mark ``$ \textcircled{\scriptsize M} $'' forwardly down the line agents towards $s_h$. Once $s_h$ detects ``$ \textcircled{\scriptsize M} $'', it emits two waves via the segment, fast ``$ \textcircled{\tiny m1} $'' and slow ``$ \textcircled{\tiny m2} $''. The fast wave ``$ \textcircled{\tiny m1} $'' moves from $p_{i}$ to $p_{i+1}$ every round, while the slow wave ``$ \textcircled{\tiny m2} $'' passes every three rounds. In this way, the fast wave ``$ \textcircled{\tiny m1} $'' bounces off $s_t$ and meets `$ \textcircled{\tiny m2} $'' at a middle agent $p_i^{\prime}$ of $S_i$ which updates label to $s_t^{\prime}$, and $p_{i+1^{\prime}}$ changes label to $s_h^{\prime}$ as well. See a demonstration in Figure \ref{fig:FastSlowWaves}. Consequently, $S_i$ is temporarily divided into two halves $S_i^1$ and $S_i^2$ labelled:
	\begin{align*}
		(\ldots, \overbrace{s_h, \ldots, s, \ldots, s_t^{\prime}}^{S_i^1}, \overbrace{s_h^{\prime},\ldots, s, \ldots, s_t}^{S_i^2},\ldots).
	\end{align*}
    Now,  $s_t^{\prime}$ emits the ``$ \textcircled{\scriptsize M} $'' mark back to $l_h$ via transmission states, from $p_i$ to $p_{i-1}$. Upon arrival of ``$ \textcircled{\scriptsize M} $'', $l_h$ invokes the sub-procedure (A) to begin collection on the first half $S_i^1$ and \textsf{Push($S$)} to move towards $s_t^{\prime}$, after which $l_h$ calls (A) again to travel into $s_t$.
		
	We argue that the line $L_i$ always has sufficient memory to store all collected arrows. The Manhattan distance will always be $\delta(l_t, s_t) > |L_i|$ if the segment $S_i$ has at least one diagonal connection. Consider the worst-case scenario of a diagonal segment in which each agent $p_i$ gains a local diagonal direction at a cost of two cardinal arrows.  Recall that each agent can store two arrows in its state, in $c_6$ and $c_7$. Given that, in the worst-case the segment contains a total of $2|S_i|$ local arrows. Thus, by applying (A) twice in each half of $S_i$, each single arrow of $S_i$ will find a room in $L_i$.
	
	We now calculate the running time of the  \textsf{CollectArrows($L_i, S_i$)} procedure on a number of rounds. Starting from steps 1 and 2 of (A), the ``\counterplay'' mark takes a journey from $l_h$ to $s_t$ requiring at most $t_1 = |L_i| + |S_i|= O(|L_i|)$ rounds. Then, the pipelined collection and rearrangement of arrows in steps 3-6,  require at most a number of parallel rounds equal asymptotically to the length of $|S_i| + |L_i|$, namely $t_2 = O(|L_i|)$. Moreover, the cost of ``\counterplay\checkmark''  transmission takes time $t_3 =|L_i|$ rounds. In (B), the propagation of  ``$ \textcircled{\scriptsize M} $'' costs $t_4 =|L_i|$, another cost $t_5 =3|S_i|$ is preserved for (1) and (2), which is the communication of fast ``$ \textcircled{\tiny m1} $'' and slow ``$ \textcircled{\tiny m2} $'' and the return of ``$ \textcircled{\scriptsize M} $'' to the head, respectively. Hence, (A) costs at most $t_A = t_1 + t_2 + t_3  = O(|L_i|)$ parallel rounds of communication, whereas (B) requires at most $t_B = t_4 + t_5  = O(|L_i|)$. The same bound holds in the worst-case by applying (A) twice. Therefore, this procedure requires a total number of at most $T = 2t_A + t_B = O(|L_i|)$ parallel rounds to draw a rout map.	 \qed
\end{proof}

\begin{figure}[hbt!]
	\centering
	\includegraphics[scale=0.8]{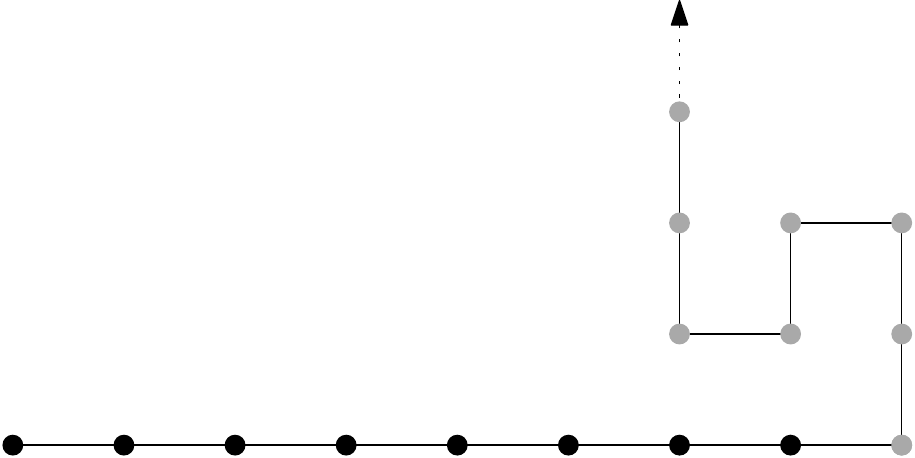}
	\caption{A configuration of $L_i$  (black dots)  and $S_i$ (grey dots).}
	\label{fig:ExCollectArrow}
\end{figure}
\begin{figure}[hbt!]
	\centering
	\includegraphics[scale=0.60]{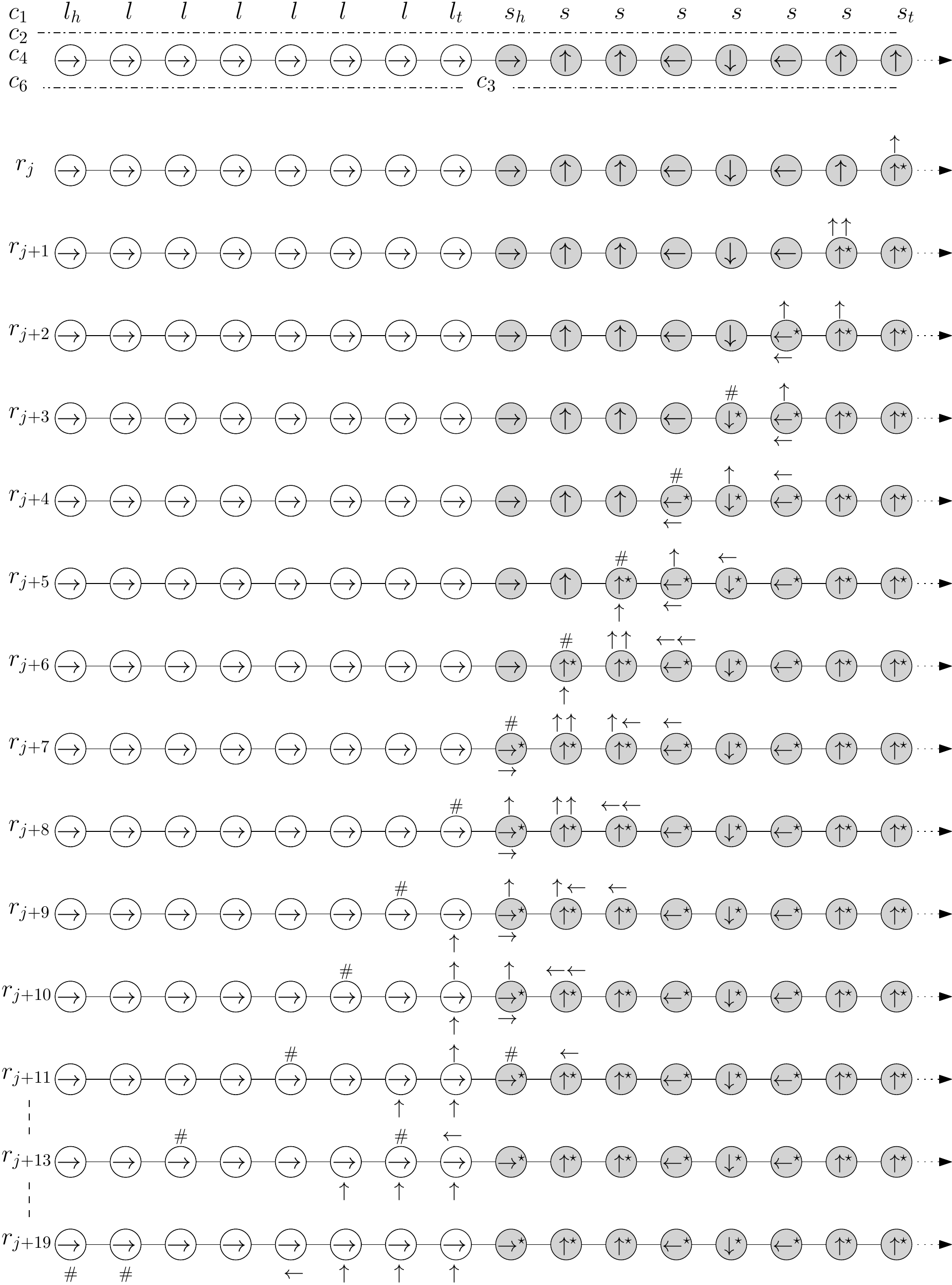}
	\caption[An implementation of the arrows collection on the shape in Figure \ref{fig:ExCollectArrow}.]{An implementation of the arrows collection on the shape in Figure \ref{fig:ExCollectArrow}. For better visibility, we represent  $L_i$ and $S_i$ as a tab, where white nodes depict $L_i$ and greys indicate $S_i$. Topmost shape shows that each agent represents its local direction $c_4$ inside nodes, label $c_1$ and transmission $c_2$ above nodes, waiting $c_3$ (only for segment agents) and map state $c_6$ (only for line agents) below nodes. The process starts from round $r_j$ downwards. See Lemma \ref{lem:CollectArrows} for a full description.}
	\label{fig:CollectArrows}
\end{figure}
\begin{figure}[hbt!]
	\centering
	\includegraphics[scale=0.85]{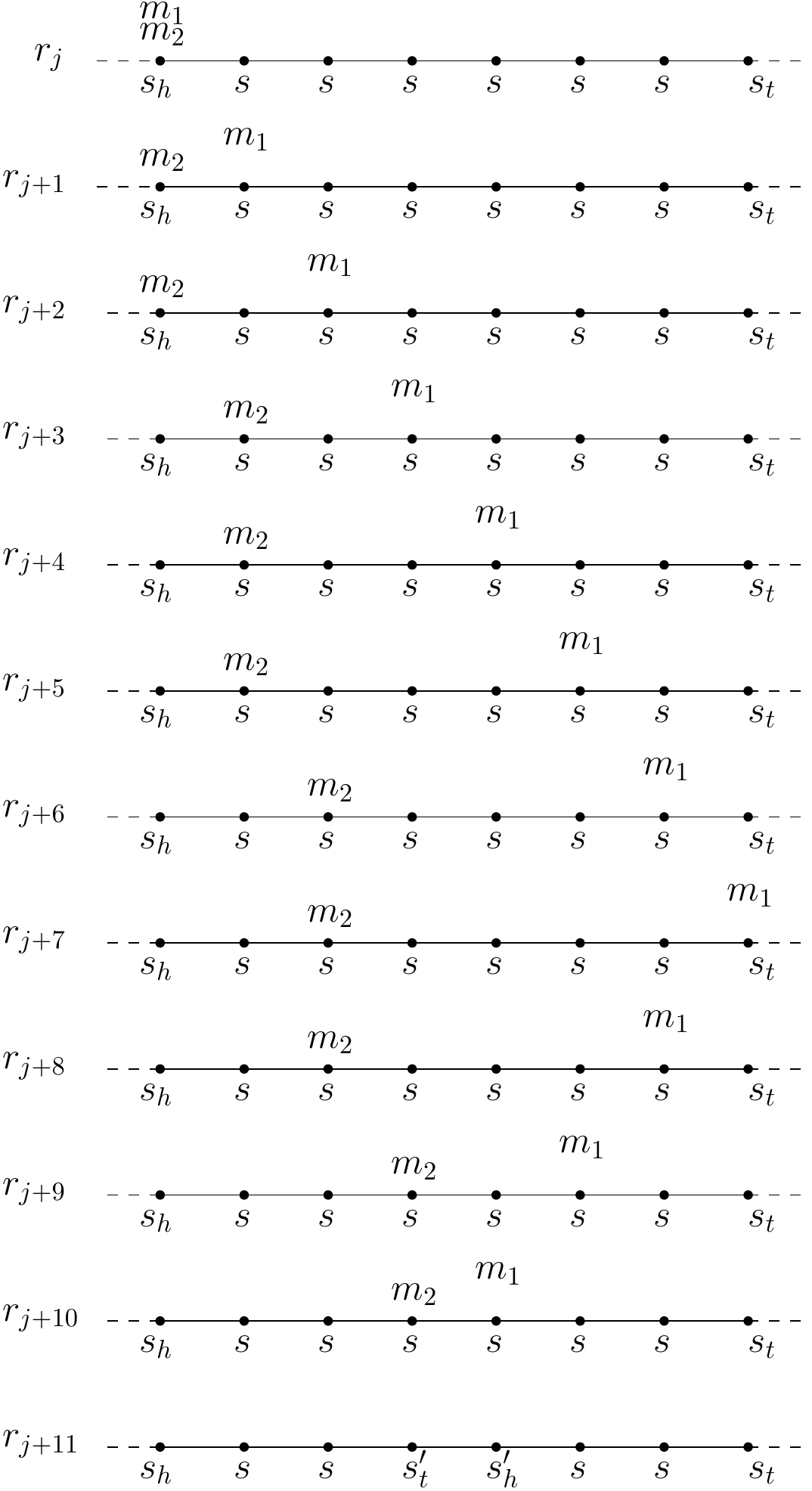}
	\caption[A fast ``$ \textcircled{\tiny m1} $'' and slow ``$ \textcircled{\tiny m2} $'' wave meeting at the middle of $S_i$ of 8 agents.]{A fast ``$ \textcircled{\tiny m1} $'' and slow ``$ \textcircled{\tiny m2} $'' wave meeting at the middle of $S_i$ of 8 agents. Observe that ``$ \textcircled{\tiny m1} $'' moves every round, while ``$ \textcircled{\tiny m2} $'' is three rounds slower.}
	\label{fig:FastSlowWaves}
\end{figure}

Finally, \textsf{ComputeDistance} and \textsf{CollectArrows} procedures completes the \textsf{DrawMap} sub-phase. By Lemmas \ref{lem:CompDis} and \ref{lem:CollectArrows}, we conclude that:

\begin{lemma} \label{lem:DrawMap}
	\textsf{DrawMap} draws a map within $O(|L_i|)$ rounds.    
\end{lemma}  

\subsection{Push the next segment $S_i$}
\label{sec:Push}

Unlike all previous sub-phases, the transformation now allows individuals to perform line movements on the grid, taking advantage of their linear-strength pushing mechanism. That is, a straight line $L_i$ of $2^i$ agents occupying a column or row of $2^i$ consecutive cells on the square grid can be pushed in a single step depending on its orientation in parallel vertically or horizontally in a single-time step. Furthermore, $L_i$ has the ability to change direction or turn from vertical to horizontal and vice versa. 

A variety of obstacles must be overcome in order to translate the global coordinator of line moves into a system of homogeneous agents capable of only local vision and communication. One of the most essential challenges is timing: an individual agent moving the line must know when to start and stop pushing. Otherwise, it may disconnect the shape and break the connectivity-preservation requirement. Further, the line may change direction and turn around while pushing; hence, it must have some kind of local synchronisation over its agents to ensure that everyone follows the same route and no one is pushed off. Failure to do so may result in a loss of connectivity, communication, or the displacement of other agents in the configuration. Moreover, pushing a line does not necessarily traverse through free space of a Hamiltonian shape; consequently, a line may walk along the remaining configuration of agents while ensuring global connectivity at the same time. However, we were able to address all of these concerns in \textsf{Push}, which will be detailed below.

After some communication, $l_h$ observes that $L_i$ is ready to move and can start \textsf{Push} now. It synchronises with $l_t$ to guide line agents during pushing. To achieve this, it propagates fast ``$ \textcircled{\tiny p1} $'' and slow ``$ \textcircled{\tiny p2} $'' marks along the line, ``$ \textcircled{\tiny p1} $'' is transmitted every round and ``$ \textcircled{\tiny p2} $'' is three rounds slower (shown early in \textsf{DrawMap}). The ``$ \textcircled{\tiny p1} $'' mark reflects at $l_t$ and meets ``$ \textcircled{\tiny p2} $'' at a middle agent $p_i$, which in turn propagates two pushing signals ``$ \textcircled{\tiny P} $'' in either directions, one towards $l_h$ and the other heading to $l_t$. This synchronisation liaises $l_h$ with $l_t$ throughout the pushing process, which starts immediately after ``$ \textcircled{\tiny P} $'' reaches both ends of the line at the same time. Recall the route map has been drawn starting from $l_t$, and hence, $l_t$ moves simultaneously with $l_h$ according to a local map direction $ \hat{a} \in A$ stored in its map component $c_6$. 

Through this synchronisation, $l_t$ checks the next cell $(x,y)$ that $L_i$ pushes towards and tells $l_h$, whether it is empty or occupied by an agent $p \not\in L_i$ in the rest of the configuration. If $(x,y)$ is empty, then $l_h$ pushes $L_i$ one step towards $(x,y)$, and all line agents shift their map arrows in $c_6$ forwardly towards $l_t$. If $(x,y)$ is occupied by $p \not\in L_i$, then $l_t$ swaps states with $p$ and tells $l_h$ to push one step. Similarly, in each round of pushing a line agent $p_i$ swaps states with $p$ until the line completely traverses the drawn route map and restores it to its original state. Figure \ref{fig:PushOccupiedCell} shows an example of pushing $L_i$ through a route of empty and occupied cells. In this way, the line agents can transparently push through a route of any configuration and leave it unchanged. Once $L_i$ has traversed completely through the route and lined up with $s_t$, then \textsf{RecursiveCall} begins. Algorithm \ref{alg:Push} provides a general procedure of \textsf{Push}.

\begin{algorithm}[hbt!] 
	\caption{\textsf{Push}} 
	\label{alg:Push} 
	Input: a straight line $L_i$ and a segment $S_i$
	\SetAlgoLined
	\DontPrintSemicolon
	\vspace{7pt} 
	
	The line head $l_h$ observes  the completion of \textsf{DrawMap}
	
	\Repeat{\upshape$l_h$ swaps labels with $s_h$}
	{
		$l_h$ emits a mark to $l_t$ to start pushing \tcp{$l_t$ sees empty or non-empty cell}
		\If{\upshape$c_6 = d_{l_t}$ point to empty cell \tcp{local arrow of $l_t$ points to empty cell}} 
		{$l_h$ syncs $L_i$: update states and push one step}
		
		\If{\upshape$c_6 = d_{l_t}$ point to non-empty cell $k$}
		{
			$l_t$ activates $k$
			
			$l_h$ syncs $L_i$  \tcp{swap and update states as described in text}
			
			$L_i$  pushes one step
		}		
	}
	\textsf{RecursiveCall} begins 
\end{algorithm}

\begin{figure}[hbt!]
	\centering
	\includegraphics[scale=0.85]{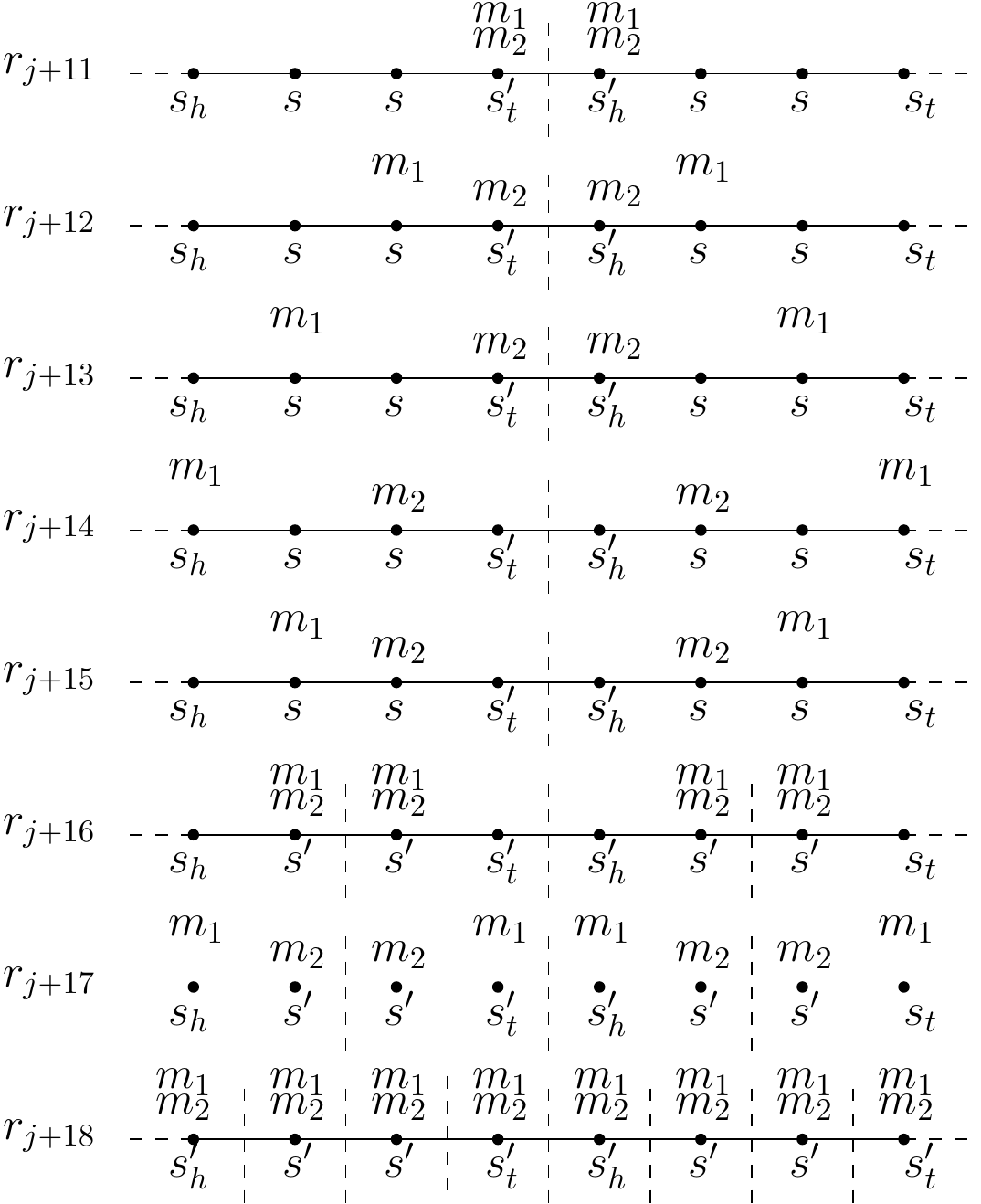}
	\caption[An implementation of synchronising 8 agents that is started in Figure \ref{fig:FastSlowWaves}.]{Synchronising 8 agents that were started in Figure \ref{fig:FastSlowWaves} where the halving procedure repeats until all agents reach a synchronised state.}
	\label{fig:FastSlowWaves_2}
\end{figure}

\subsubsection{Agents synchronisation}
Many agent behaviours, including state swapping and line movements (parallel pushing), are realised to be very efficient in the centralised systems of a global coordinator. In contrast, the constraints in this model make these simple tasks difficult, as individuals with limited knowledge cannot keep track of others during the transformation. This may result in the disconnection of the whole shape, a modification in the rest of the configuration or even the loss of a chain of actions that halts the transformation process. However, the synchronisation of agents can assist to tackle such an issue where individuals can organise themselves to eventually arrive at a state in which all of them conduct tasks concurrently. This concept is similar to a well-known problem in cellular automata known as the firing squad synchronisation problem, which was proposed by Myhill in 1957. McCarthy and Minsky provided a first solution to this problem \cite{minsky67}. The following lemma demonstrates how their solution can be translated to our model in order to coincide a Hamiltonian path of $n$ agents  in such a way that they can perform concurrent actions in linear time.  

\begin{lemma}[Agents synchronisation] \label{lem:synchronisation}
	Let $P$ denote a a Hamiltonian path of $n$ agents on the square grid, starting from a head $p_1$ and ending at a tail $p_n$, where $p_1 \ne p_n$. Then, all agents of $P$ can be synchronised in at most $O(n)$ rounds.
\end{lemma}
\begin{proof}
	From \cite{minsky67}, the strategy consists of two cases, even and odd number of agents. First, the head $p_1$ emits fast mark ``$ \textcircled{\tiny m1} $'' and slow mark ``$ \textcircled{\tiny m2} $'' towards the tail $p_n$. The  ``$ \textcircled{\tiny m1} $'' mark is communicated from $p_i$ to $p_{i+1}$ via transmission components in each round, while is transmitted from $p_i$ to $p_{i+1}$ every three rounds. When ``$ \textcircled{\tiny m1} $'' reaches the other end of the path $p_n$, it returns to $p_1$. Thus, the two marks collide exactly in the middle (see an example in Figure \ref{fig:FastSlowWaves}). Now, the two agents who witness the collision update to a special state, which will effectively split $P$ into two sub-paths. Both agents repeat the same procedure in each half of length $n/2$ in either direction of $P$. Repeat this halving until all agents reach a special state (collision witness) in which they all perform an action simultaneously. An implementation of this synchronisation is depicted in Figure \ref{fig:FastSlowWaves_2}.
	
	Assume a path $P$ of  $n$ odd agents in which $p_1$ emits, ``$ \textcircled{\tiny p1} $'' and ``$ \textcircled{\tiny p2} $'' along  $P$. In this case, the two marks meet in a slightly different way, at an exact single middle agent $p_i$ on $P$. This agent  $p_i$  observes a predecessor $p_{i-1}$ showing ``$ \textcircled{\tiny m2} $'' and successor  $p_{i+1}$ showing ``$ \textcircled{\tiny m1} $''  in transmission state and responds by switching into another special state that allows it to play two roles. That is, it emits  ``$ \textcircled{\tiny p1} $'' and ``$ \textcircled{\tiny p2} $'' to both directions of $P$, this effectively splits $P$ into two sub-paths of length $n/2 -1$ each. Now, repeat the process in each half until the two marks intersect in the middle, at which point two agents notice the collision and change to a special state. In the same way, divide until all agents have updated to a synchronised state. Figure \ref{fig:FastSlowWaves_3} depicts the synchronisation in the odd case.
	
	Now, we are ready to describe the state transitions. In the first round $r_j$,   $p_1$ updates to $p_1.c_2 \gets \textcircled{\tiny m1}$ and combines ``$ \textcircled{\tiny m2} $''  with ``$w$'' in waiting state,  $p_1.c_3 \gets \textcircled{\tiny m2}w$. Next round $r_{j+1}$, $p_1$ updates state to $p_1.c_3 \gets \textcircled{\tiny m2}$ and $p_1.c_2 \gets \cdot $. In the third round $r_{j+2}$, $p_1$ updates transmission state to $p_1.c_2 \gets \textcircled{\tiny m2} $. Whenever $p_i$ notices: (1) $p_{i-1}$ (or $p_{i+1}$) showing ``$ \textcircled{\tiny m1} $'', $p_i$ shifts transmission to $p_i.c_2 \gets \textcircled{\tiny m1}$ and $p_{i-1}$ o(r $p_{i+1}$) rests their transmission next round. (2) $p_{i-1}$ (or $p_{i+1}$) showing  ``$ \textcircled{\tiny m2} $'', $p_i$ updates waiting state to $p_i.c_3 \gets \textcircled{\tiny m2}w$ and $p_{i-1}$ (or $p_{i+1}$) rests their transmission next round. (3) $p_{i+1}$ showing ``$ \textcircled{\tiny m1} $''  and $p_{i-1}$ presenting ``$ \textcircled{\tiny m2} $''  (or vice versa), $p_i$ updates to another special state and repeats (1). When both $p_i$ and $p_{i+1}$ are presenting ``$ \textcircled{\tiny m1} $'' and  ``$ \textcircled{\tiny m2} $'', they update into a special state and repeat the procedure of $p_1$ in either directions. Repeat until all agents and their neighbours reach a special state where all are synchronised. 
	
	Let us now analyse the runtime of this synchronisation in a number of rounds. The fast mark ``$ \textcircled{\tiny m1} $'' moves along $P$ taking $n$ rounds plus $n/2$ to walks back to the centre in a total of at most $3n/2$ rounds. The same bound applies to the slow mark ``$ \textcircled{\tiny m2} $'' arriving and meeting ``$ \textcircled{\tiny m1} $'' in the middle. The whole procedure is now repeated on the two halves of length $n/2$, each takes $3n/4$ rounds. This adds up to a total $\sum_{i=1}^{n} 3n/2^i = 3n/2 + 3n/4 +\ldots+0 = 3n(1/2+1/4 + \ldots+0)= 3n (1) = 3n $. Therefore, this synchronisation requires at most $O(n)$ rounds of communication.  \qed
\end{proof}
\begin{figure}[hbt!]
\centering
\includegraphics[scale=0.85]{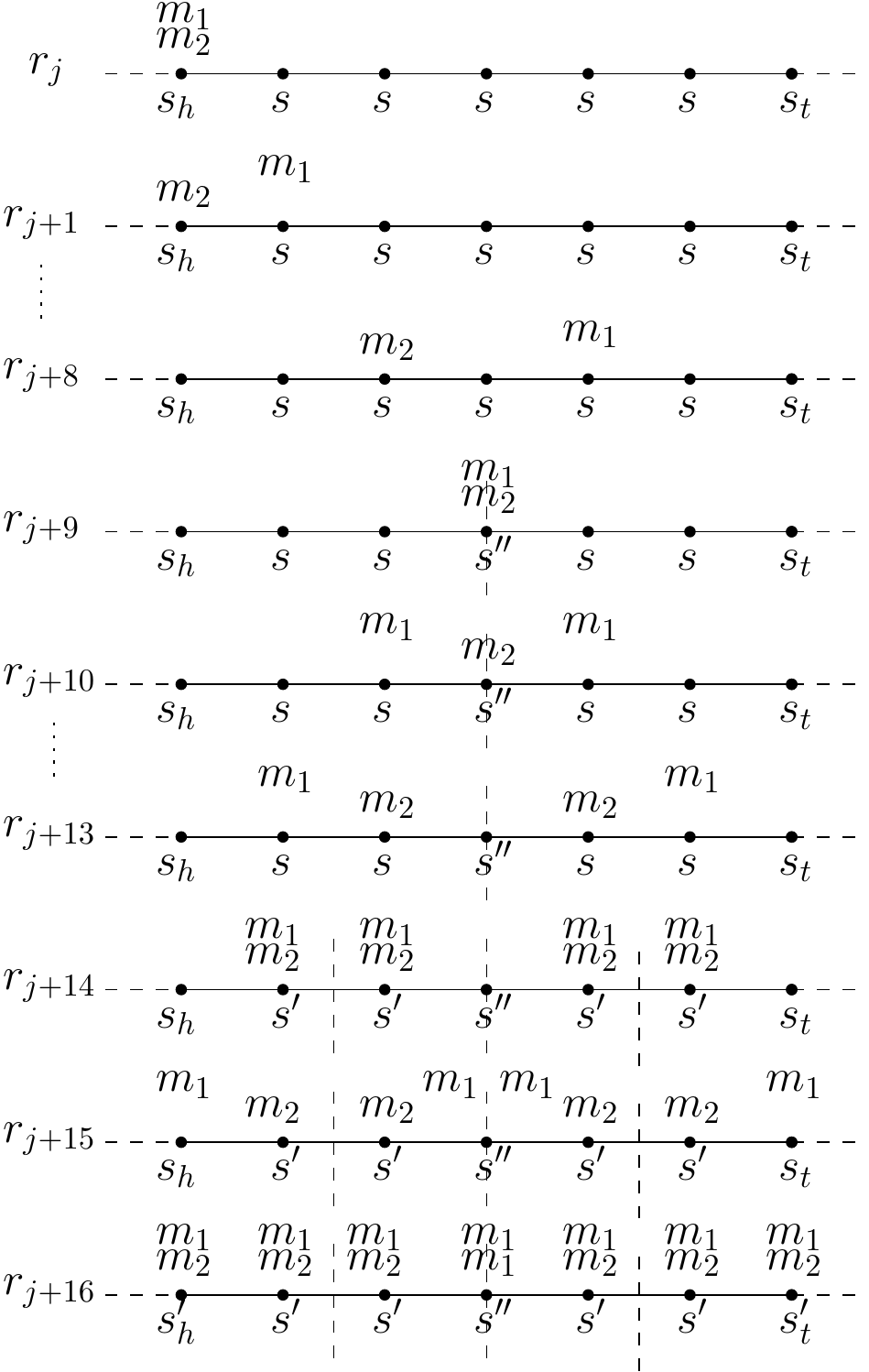}
\caption{An example of synchronising 7 agents - odd case.}
\label{fig:FastSlowWaves_3}
\end{figure}

Now, we show that under this model a number of consecutive agents forming a straight line $L_i$, can traverse  transparently through a route $R$ of cells on the grid of any configuration $C_R$ exploiting only their local knowledge, without breaking the connectivity of the whole shape.

\begin{lemma} \label{lem:PushCorrectness}
	Let $L_i $ denote a terminal straight line and $R$ be a rectangular path of any configuration $C_R$, starting from a cell adjacent to the tail of $L_i$, where $R \le 2|L_i|-1$. Then, there exists a distributed way to push $L_i$ along $R$ without breaking connectivity.  
\end{lemma}
\begin{proof}
	In Algorithm \ref{alg:Push}, the line head $l_h$ observes the collection mark ``\counterplay\checkmark'' indicating the completion of \textsf{DrawMap($S$)}, which draws a route $R$ (see Definition \ref{def:route}). As a result, $l_h$ emits the question mark ``$?$'' to $l_t$, which will broadcast via line agent transmission states from $p_i$ to $p_{i+1}$. Once `$?$'' arrives there, $l_t$ checks whether its map arrow $d_{l_t}$ points to an empty or occupied cell, and if so, it emits a special mark ``$ \textcircled{\tiny Y} $'' back to $l_h$ indicating that a rout is free to push. By an application of Lemma \ref{lem:synchronisation},  $l_h$ synchronises all line agents to reach a concurrent state in which the following actions occur concurrently: (1) $l_h$ pushes $L_i$ one position towards $l_t$, based on its local direction on push state $c_7$. (2) $l_t$ pushes one position based on its map arrow $c_6$ either in line direction or perpendicular to $L_i$. In the latter, $l_t$ updates state to $c_4 \gets c_6$ and tells predecessor to turn next round. In general, (3) If $p_i$ turns, it updates local direction $c_4 \gets c_6$, and  $p_{i-1}$ updates push component $p_{i-1}.c_7\gets p_i.c_6$.  (4) $p_i$ of a present push component $c_7$ moves one step in the direction held in $c_7$, which then rests to $p_i.c_7 \gets \cdot$. (5) All line agents shift local map direction forwardly towards $l_t$,  $p_i.c_6 \gets p_{i-1}.c_6$. Repeat these transitions until $l_t$ encounters the segment tail $s_t$ on the route through which $l_t$ tells $l_h$ to sync and push again, while $l_t$ and $s_t$ swaps their states. Hence, any $p_i$ meets $s_t$, they swap states and rest their $c_6$. Eventually, $l_h$ stops pushing once it meets and swaps states with $s_t$. An example is shown in Figure \ref{fig:Push_L_R_empty}.
		\begin{figure}[hbt!]
		\centering
		\includegraphics[scale=0.725]{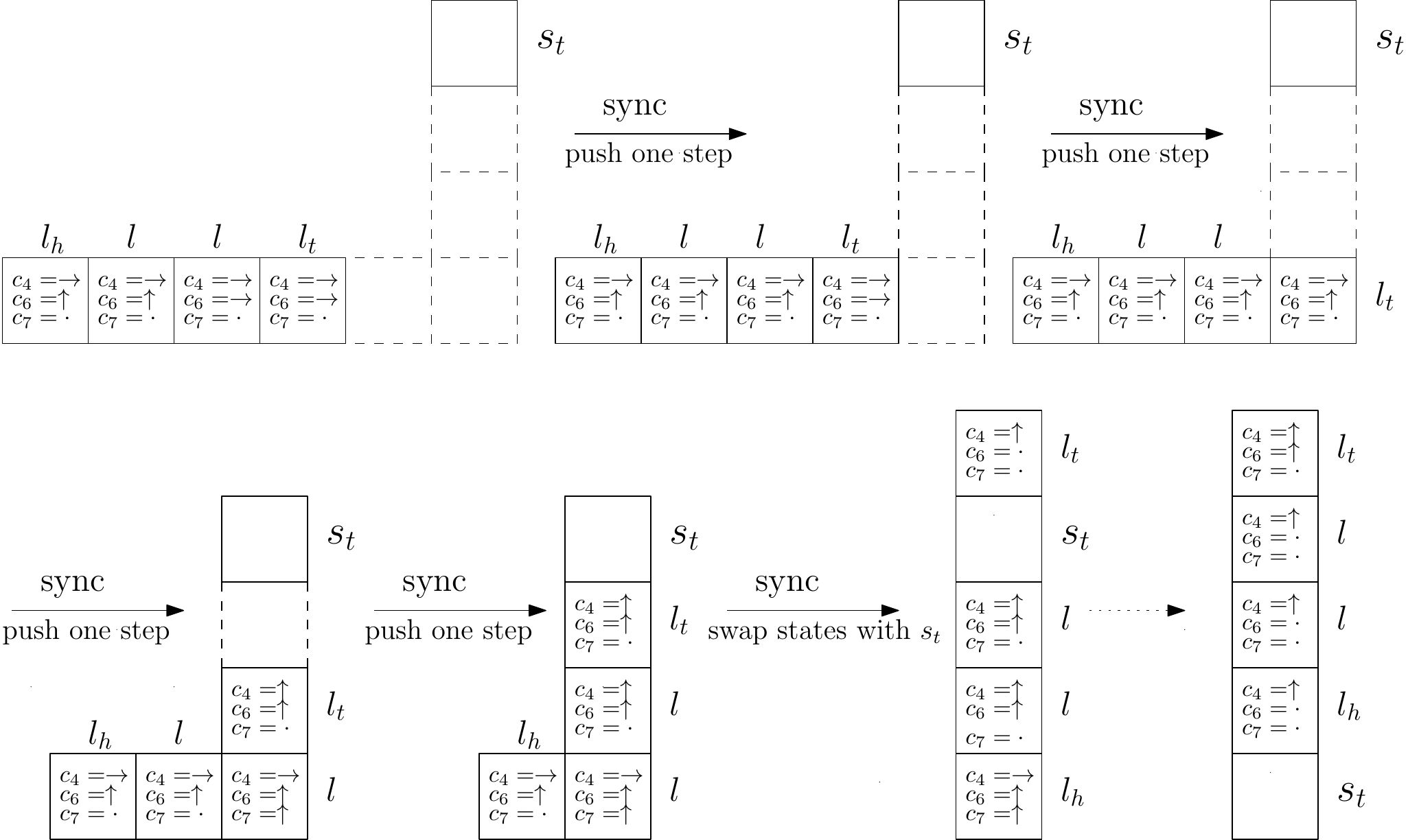}
		\caption[A line $L_i$ of four agents pushing through a route of empty cells.] {A line $L_i$ of four agents pushing through a route of empty cells towards $s_t$. All affected states ($c_4, c_5$ and $c_7$) are shown inside each occupied cell.}
		\label{fig:Push_L_R_empty}
	\end{figure}
	
	During pushing through an L-shape route $R$, $L_i$ may turn one or at most three times. In the following, we show that the number of turns depends on the orientation of both $L_i$ and $R$. Without loss of generality, assume a horizontal $L_i$ turning at a corner towards $s_t$, such as Figure \ref{fig:Push_L_R_empty} where $L_i$ will temporally divide into two perpendicular sub-lines while traversing to $s_t$. By a careful application of Lemma \ref{lem:synchronisation}, both can be synchronised and organised to perform two parallel pushing where $l_h$ liaises with $l_t$ and push the two perpendicular sub-lines concurrently. Now, assume $s_t$ is placed two cells above the middle of $L_i$, resulting in a route $R$ of three turns along which $L_i$ temporally transforms into three perpendicular sub-lines. Three agents  simultaneously drive everyone to advance one step ahead on $R$. Therefore, the line can be synchronised to perform at most three parallel pushing operations that are asymptotically equivalent to the cost of one pushing, without breaking connectivity.  Below are transitions that demonstrate how $L_i$ pushes along $R$ while satisfying all of the transparency properties of line moves in \cite{AMP2020}:     
	\begin{itemize}
		\item[--] No delay: $L_i$ traverses $R$ of any configuration $C_R$ within the same asymptotic number of moves, regardless of how dense is  $C_R$.
		
		\item[--] No effect: $L_i$ restores all occupied cell to their original state and keeps $C_R$ unchanged after traversing $R$.
		
		\item[--] No break: $L_i$ preserves connectivity while traversing along $R$. 
	\end{itemize} 
	
	Now, assume $L_i$ walks over a route $R$ of non-empty cells occupied by other agents (denoted by $k$) in the configuration that are not on $S_i$. Whenever $L_i$ walks through $R$ and $l_t$ meets $k$ on $R$, $l_t$ tells $l_h$ to stop pushing. The agent $k$ now updates to a temporary state labelled $k_{l_t}$ if it has a similar arrow of $l_t$ or $k_{l_t}c$ if it has a turn. Based on the map arrow of $l_t$, $k_{l_t}$ acts as a tail and checks whether the next cell $(x,y)$ on $R$ is empty, which accordingly triggers to one of the following states : (1)  $(x,y)$ is empty, then $k_{l_t}$ emits a mark back to $l_h$ to sync and push $L_i$ one step further. During this, $k_{l_t}$ changes state to $k_{l}$ and each synchronised agent $p_{i}$ shifts map arrow to $p_{i+1}$. During pushing, $k_{l}$ swaps states with its predecessor and ensures that it remains in the same position (see Figure \ref{fig:PushEmptyCell}) until it meets $l_h$, at which $k_{l}$ can update to its original state $k$. (2) $(x,y)$ is occupied by another agent labelled $k$, then (a) $k_{l_t}$ changes to $k_{l_t}^{\star}$ (or $k_{l_t}^{\star}c$ if the map arrow is a turn) and $k$ into $k_{l_t}$, and  (b) $k_{l_t}$ emits a special mark to the line agents asking for the next map arrow $d_{k_{l_t}}$. Repeat this process as long as $d_{k_{l_t}}$ indicates an occupied cell. Once $k_{l_t}$ observes an empty cell $(x^{\prime},y^{\prime})$, it performs (1)  and updates $k_{l_t}^{\star}$. See a demonstration in Figure \ref{fig:PushOccupiedCell}.
	
	\begin{figure}[hbt!]
		\centering
		\includegraphics[scale=0.75]{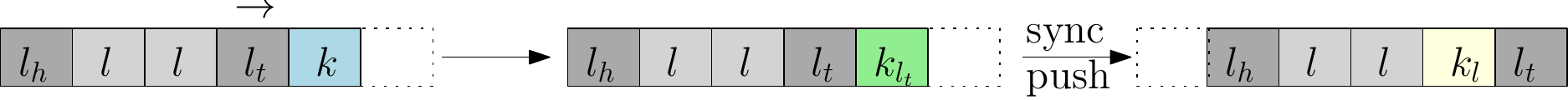}
		\caption[A line pushing through a non-empty cell.]{A line $L_i$ of agents within grey cells pushing through a non-empty cell in blue with a right map direction above $l_t$.}
		\label{fig:PushEmptyCell}
	\end{figure}
	\begin{figure}[hbt!]
		\centering
		\includegraphics[scale=0.85]{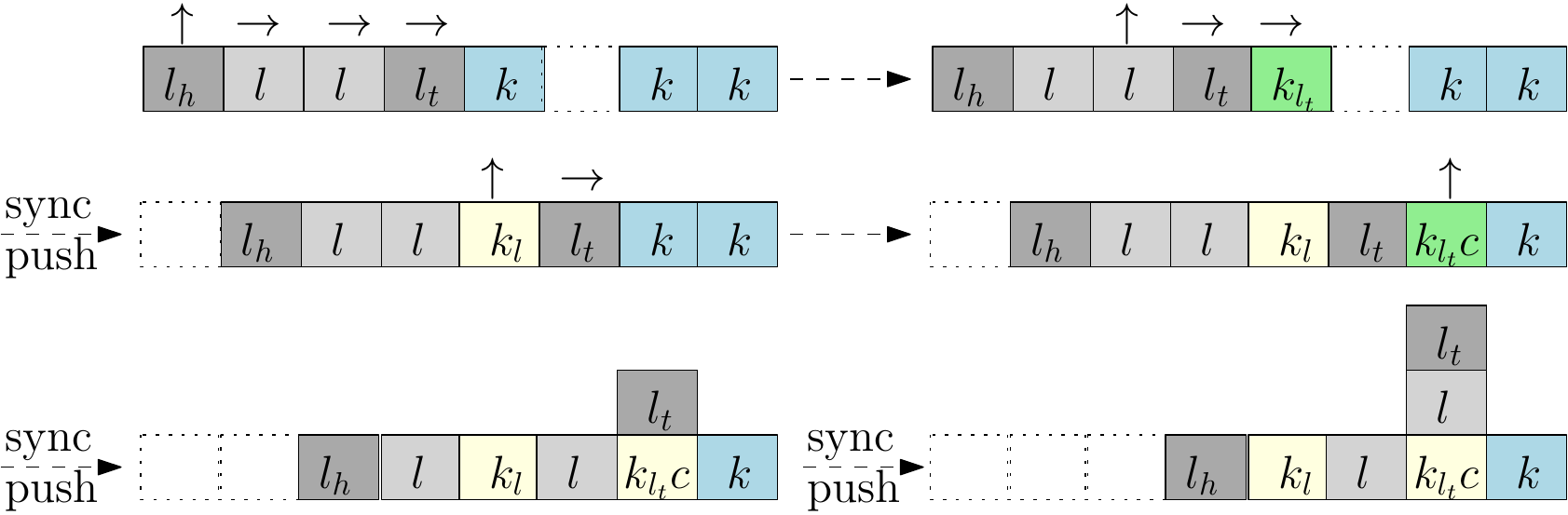}
		\caption[A line pushing and turning through empty and non-empty cell.]{A line $L_i$ of agents inside grey cells, with map directions above, pushing and turning through empty and non-empty cells in blue (of label $k$). The green and yellow cells show state swapping.}
		\label{fig:PushOccupiedCell}
	\end{figure}
	
	When $L_i$ moves through a series of non-empty cells, it guarantees that they are neither separated or disconnected while pushing. To achieve this, when $l_t$ or $k_{l_t}$ calls for synchronisation, any line agent $p_i$ labelled $l$ whose successor shows a label with star ($k_{l}^{\star}$ or $k_{l}^{\star}c$ ), both swap their states. It continues to swap states forwardly via consecutive non-empty cells until reaches the tail $l_t$ or a line agent $l$. Though, when $L_i$ traverses entirely through $R$ and reaches the segment tail $s_t$, it may find another non-empty cell after swapping states with $s_t$. Hence, the same argument above still holds in this case. Figure \ref{fig:Corner_Pushing} shows this case.
	
	\begin{figure}[hbt!]
		\centering
		\includegraphics[scale=0.7]{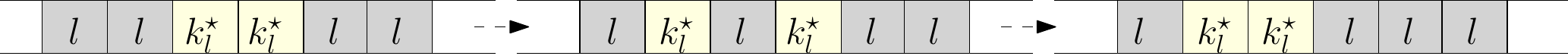}
		\caption[Four agents of a line swap states with other agents.]{Four agents in a line inside grey cells swap states with others occupying consecutive yellow cells.}
		\label{fig:ForwardAgents}
	\end{figure}
	
	Whenever $L_i$ pushes into an empty cell $(x,y)$, it fills $(x,y)$ with an agent $p \in L_i$. During pushing, $L_i$ always keeps the original position of a non-empty cell and restores it to its initial state (via state swapping). However, there exists a case that may break the connectivity. Consider a line $L_i$ pushing along $R$ and turning at a corner agent labelled $k_{l_t}c$, which has two diagonal neighbours where both are not adjacent to any line agent, as depicted in Figure \ref{fig:Corner_Pushing} top. In this case, when $k_{l_t}c$ moves down, it will break connectivity with its upper diagonal neighbour. Hence, the transformation resolves this issue locally depending on the agents' local view. When $k_{l_t}c$ observes a pushing agent and has one or two diagonal neighbours, it temporarily switches to a state that allows it to move one step further while $l_t$ updates into a turning agent. This also permits all line agents to turn sequentially until they reach the head $l_h$, which turns and waits for $k_{l_t}c$ to return to its initial cell. Figure \ref{fig:Corner_Pushing} depicts how to handle this situation. Other orientations follow symmetrically by rotating the system $90^{\circ} , 180^{\circ}$ or $270^{\circ}$ clockwise and counter-clockwise.  
	
	\begin{figure}[hbt!]
		\centering
		\includegraphics[scale=0.70]{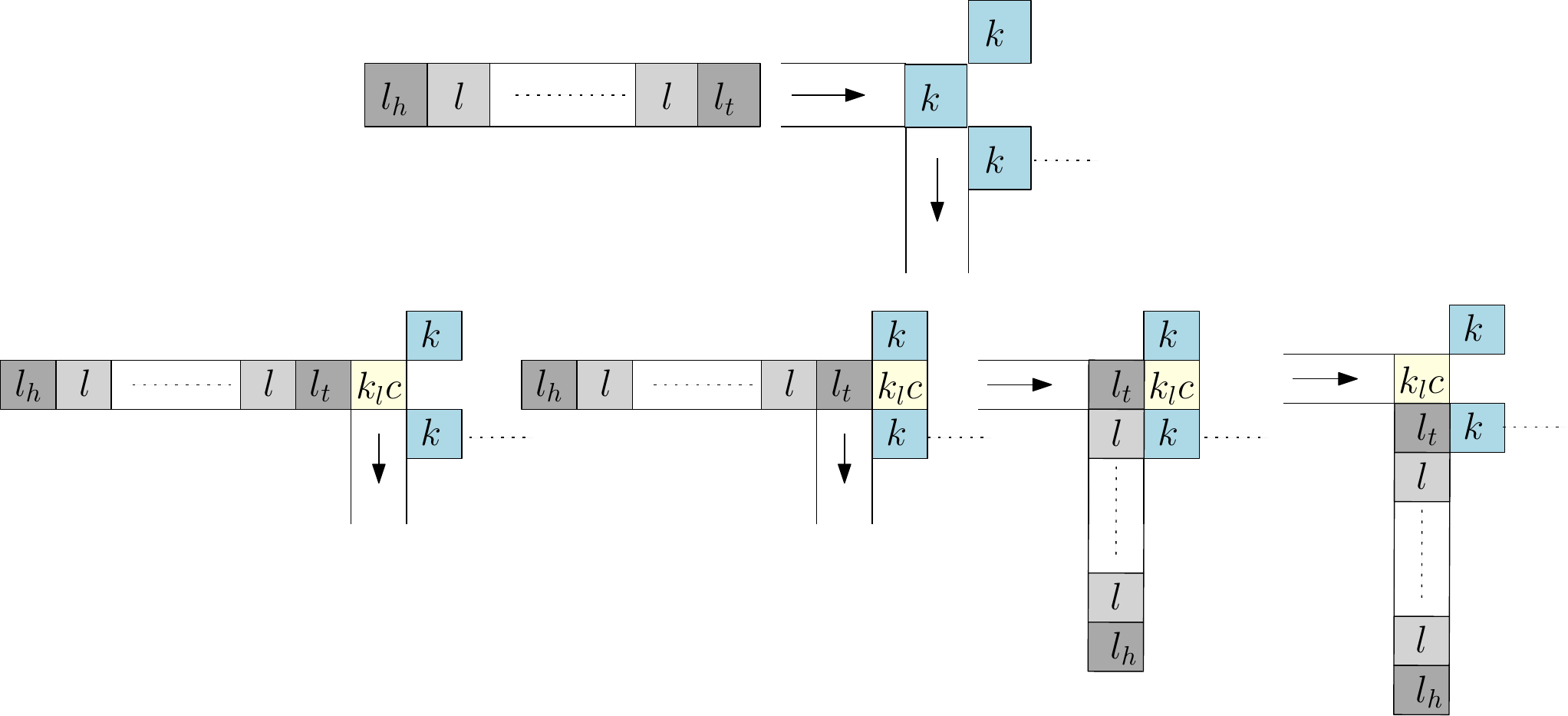}
		\caption[A line pushing through a disconnecting corner.]{A line $L_i$ pushing through a route $R$ turns at a corner agent labelled $k_{l_t}c$ that has two diagonal neighbours, neither of which is adjacent to any line agent.}
		\label{fig:Corner_Pushing}
	\end{figure}
	
	Thus, all agents of  $L_i$ are labelled and organised in such a way that can transparently push through a route $R$ of any configuration $C_R$, whether it is being empty or partially/fully occupied.  It implies that $L_i$ remains connected when travelling as well as the whole configuration. Further, the original state of $C_R$ has been restored and all of its nodes (if any) have been left unchanged. As a result, $L_i$ meets all of the transparency criteria of line moves in \cite{AMP2020}. \qed
\end{proof}

The complexity of \textsf{Push} is provided in the following lemma based on the number of line moves and communication rounds.

\begin{lemma} \label{lem:PushRuntime}
	A straight line $L_i$ traverses through a route $R$ of any configuration $C_R$, taking at most $O(|L_i|)$  line moves within  $O(|L_i|\cdot|R|)$ rounds.
\end{lemma} 
\begin{proof}
	The bound of moves depends on three factors, the number of empty cells on $R$, the length of $L_i$  and the number of turns on $R$. Say that $R$ is free of agents (fully empty) and has at most 3 turns, then $L_i$ requires at most $|L_i| + 3|L_i| + |L_i| = 5 |L_i|= O(|L_i|)$ moves (proved in Lemma \ref{lem:PushCorrectness}) to push through $R$.  On the other hand, the communication cost of this sub-phase could be very high in the case of a fully occupied route $R$ when individuals perform many functions such as synchronisation, activation, state swapping, and map arrow forwarding. Those actions can be carried either sequentially or concurrently during the transformation and can be analysed independently of each other. In this case, we set an upper bound on the most dominating work.
	
	Assume that $R$ is completely occupied by other agents in the shape (in a worst-case), from the cell adjacent to the line tail $l_t$ tp the cell adjacent to $s_h$. Then, $l_t$ needs to traverse over at most $ |R|$ agents in order to arrive at $s_h$, which costs $t_1^c = |R|$ rounds. Further, $l_t$ requires a number of synchronisations equal to $|L_i|$  to move all line agents along $R$ at a cost of no more than $t_2^c=|L_i|\cdot|R|$ rounds. In each synchronisation, a line agent swaps its state with $ |R|$ agents and forwards its map direction over line agents to $l_t$ within at most $t_3^c=|L_i|+|R|$ rounds.  Thus, this sub-phase results in a maximum number of communications $T^c = t_1^c  +t_2^c +t_3^c = |R| + (|L_i|\cdot|R|) + (|L_i|+|R|) = O(L_i|\cdot|R|) $ rounds. This bound holds when other agents occupy $|L_i|$ consecutive horizontal and vertical cells beyond $s_h$.  \qed
\end{proof}

\subsection{Recursive call on the segment $S_i$ into a line $L_i^{\prime}$}
\label{sec:RecursiveCall}

This sub-phase, \textsf{RecursiveCall}, is the heart of this transformation and is recursively called on the next segment $S_i$, which eventually transforms into another straight line $L_i^{\prime}$ of $2^i$ agents.

When a segment tail $s_t$ swaps states with $l_h$, it accordingly acts as follows: (1) propagates a special mark transmitted along all segment agents towards the head $s_h$, (2) deactivates itself by updating label to $c_1 \gets k$, (3) resets all of its components, except local direction in $c_4$. Similarly, once a segment agent $p_i$ observes this special mark, it propagates it to its successor $p_{i+1}$, deactivates itself, and keeps its local direction in $c_4$ while resetting all other components. When the segment head $s_h$ notices this special mark, it changes to a line head state ($c_1 \gets l_h$) and then recursively repeats the whole transformation from round $1$ to $i-1$. Figure \ref{fig:DignalPattern} presents a graphical illustration of \textsf{RecursiveCall} applied on a diagonal line shape. 

\subsection{Merge the two lines $L_i$ and $L_i^{\prime}$}
\label{sec:Merge}

The final sub-phase of this transformation is \textsf{Merge}, which combines two straight lines into a single double-sized line, described as follows. The previous sub-phase, \textsf{RecursiveCall}, transforms the segment $S_i$ into a straight line $L_i^{\prime}$, starts from a head $l_h$ and ends at a tail $l_t$. Currently, the tail of $L_i^{\prime}$ occupies a cell adjacent to the head of $ L_i$. Hence, $l_h$ can simply check if $L_i^{\prime}$ is in line or perpendicular to $ L_i$ exploiting the previous procedure of \textsf{CheckSeg}. Without loss of generality, say that the tail agent of $L_i^{\prime}$ occupies cell $(x,y)$ and $L_i$ occupies cells $(x,y), \ldots, (x+|L_i|-1, y)$. Then, $L_i^{\prime}$ could be either (1) perpendicular with agents occupying $(x, y), \ldots, (x, y+|L_i|-1)$  or (2)  in line on cells $(x,y), \ldots, (x-|L_i|-1, y)$. In (1), $l_h$ emits a mark that  travels via agents of $ L_i^{\prime}$ until it reaches the other head, where it asks to change the direction of $ L_i^{\prime}$, allowing  $L_i^{\prime}$ and $L_i$ to combine into a single straight line $L_{i+1}$ of double length and designate one head and tail for $L_{i+1}$.  In (2),  $L_i^{\prime}$ and $L_i$ have already formed $L_{i+1}$; all that remains is to switch and update labels to assign a head $l_h$, tail $l_t$ and $2^{i+1}$ line agents $l$ in between.

Now, it is sufficient to upper bound this sub-phase by analysing only a worst-case of (1). Obviously, the straight line $ L_i^{\prime} $ pushes and turns within a distance equal to its length in order to line up with $L_i$. It is worth noting that the agents of $ L_i^{\prime} $ do not require full synchronisation for each push. Instead, they simply need to sync the head and tail of $ L_i^{\prime} $ where both perform pushing at the same time. When an agent $p_i \in L_i^{\prime}$ turns, it tells its predecessor $p_{i-1} \in L_i^{\prime}$ to turn too. Hence, the total number of moves is at most $O(|L_i^{\prime}|)$. The communication cost splits into: (1) A special mark from $l_h$ traverses across  $ L_i^{\prime}$ in $O(|L_i^{\prime}|)$ rounds. (2) All agents of $ L_i^{\prime} $ synchronise in $O(|L_i^{\prime}|)$ rounds. (3) Label swapping costs at most  $|L_i^{\prime}| + |L_i|= O(|L_i^{\prime}|)$.  Therefore, all agents in \textsf{Merge} communicate in linear time, and then we can say:

\begin{lemma} \label{lem:Merge}
	An execution of \textsf{Merge} requires at most $O(|L_i|)$ line moves and $O(|L_i|)$ rounds of communication.
\end{lemma} 

Finally, we analyse the recursion in a worst-case shape in which individuals consume their maximum energy to communicate and move. The runtime is based on the analysis of the centralised version that has been proved in Section \ref{sec:Hamiltonianshapes}. Let $T_i^c$ and $T_i^m$ denote the total number of  communication rounds and moves in phase $i$, respectively, for all $i \in {1, \ldots , \log n}$. Apart from \textsf{RecursiveCall}, the $2^{i}$ agents forming a straight line $L_i$  in phase $i$ go through \textsf{DefineSeg}, \textsf{CheckSeg}, \textsf{DrawMap},  \textsf{Push} and \textsf{Merge} sub-phases that take total parallel rounds of communication $t^{c}_i $ at most:
\begin{align*}
	t^{c}_i  = (4\cdot|L_i|) + (|L_i| \cdot |R|) \approx O(|L_i| \cdot |L_i|).
\end{align*}

Then, in \textsf{Push} and \textsf{Merge} sub-phases, the line $L_i$ traverses along a route of  total movements $t^{m}_i $ in at most:
\begin{align*}
	t^{m}_i  = |L_i| + |L_i^{\prime}| =  O(|L_i|).
\end{align*}

Now, let $T^{c}_{i-1}$ denote a total number of  parallel rounds required for a recursive call of \textsf{RecursiveCall} on $2^{i}$ agents of the segment $S_i$, which transforms into another straight line $L^{\prime}_i$. Given $|L_i| = 2^i$, this recursion in phase $i$ costs a total rounds bounded by:
\begin{align*}
	T^{c}_{i}  &\le i \cdot (|L_i| \cdot |L_i|) \le i \cdot (2^i)^2 \\
	T^{c}_{i}  & O(\le  i \cdot n^2).
\end{align*}	

Thus, we conclude that the call of \textsf{RecursiveCall} in the final phase $ i =\log n $ requires a total rounds $T^{c}_{\log n}$:
\begin{align*}
	T^{c}_{\log n} &\le n^2 \cdot \log n\\
	&= O(n^2\log n ).
\end{align*}

The same argument follows on the total number of movements $T^{c}_{i-1}$ for a recursive call of \textsf{RecursiveCall}, which costs at most:  
\begin{align*}
	T^{m}_{i}  &\le i \cdot |L_i| \le i \cdot (2^i)\\
	T^{m}_{i}  &\le  O(\le  i \cdot n).
\end{align*}	

Finally, by the final phase $ i =\log n $, all agents in the system pushes a total number of moves   $T^{m}_{\log n}$ that bounded by: 
\begin{align*}
	T^{m}_{\log n} &\le n \cdot \log n  \\
	&= O(n\log n).
\end{align*}

Overall, given a Hamiltonian path in an initial connected shape $S_I$ of individuals of limited knowledge and permissible line moves, the following lemma states $S_I$ can be transformed into a straight line $S_L$ in a number of moves that match the optimal centralised transformation fulfilling the connectivity-preserving condition.  

\begin{lemma} \label{lem:HamiltonToLine}
	Given an initial  Hamiltonian shape $S_I$ of $n$ agents, this strategy transforms $S_I$ into a straight line $S_L$ of the same order in $ O(n\log n) $ line moves and $ O(n^2\log n) $ rounds, while preserving connectivity during transformation.
\end{lemma} 

Thus, we can finally provide the following theorem:

\begin{theorem} \label{theo:Walk_Through_1}
	The above distributed transformation solves {\sc HamiltonianLine} and takes at most $ O(n \log_2 n) $ line moves and $ O(n^2 \log_2 n) $ rounds.
\end{theorem}
\newpage

%
%
%
\bibliographystyle{splncs04}
\bibliography{ref}
\end{document}